\documentclass[reprint,aps,physrev,superscriptaddress,10pt,notitlepage,prd,nofootinbib,onecolumn]{revtex4-2}
\usepackage{hyperref} %Automatically links \label and \ref commands; Always load last
\hypersetup{
    colorlinks=true,       % false: boxed links; true: colored links
    linkcolor=red,          % color of internal links
    citecolor=blue,        % color of links to bibliography
    filecolor=magenta,      % color of file links
    urlcolor=blue           % color of external links
}
\usepackage{physics} %Useful commands such as \mqty
\usepackage{tensor} %Useful for writing tensor indices
\usepackage[dvipsnames]{xcolor}
\usepackage{amsthm}
\usepackage{amssymb}
\usepackage{framed}
\usepackage{graphicx}

\newtheorem{theorem}{Theorem}

\newcommand{\rd}{{\rm d}}

\newcommand{\fref}[1]{figure~\ref{#1}}
\newcommand{\Fref}[1]{Figure~\ref{#1}}
\newcommand{\tref}[1]{table~\ref{#1}}
\newcommand{\Tref}[1]{Table~\ref{#1}}
\newcommand{\sref}[1]{section~\ref{#1}}
\newcommand{\Sref}[1]{Section~\ref{#1}}
\newcommand{\aref}[1]{appendix~\ref{#1}}

\newcommand{\Revise}[1]{{#1}}

\newcommand{\Comment}[1]{
}

\begin{document}

%\arxivnumber{1234.56789} % if you have one

%\title{Dynamical nonlinear tails of scalar perturbations on a Schwarzschild background}
\title{Dynamical nonlinear tails in Schwarzschild black hole ringdown}

\author{Siyang Ling, Sabeela Shah, Sam S. C. Wong}

\affiliation{City University of Hong Kong,\\
Tat Chee Avenue, Kowloon, Hong Kong SAR, China}

% E-mail addresses: only for the corresponding author
\email{siyaling@cityu.edu.hk}
\email{sabeeshah2-c@my.cityu.edu.hk}
\email{samwong@cityu.edu.hk}

% \abstract{
\begin{abstract}
  Nonlinear tails in black hole perturbations, arising from second-order effects, present a distinct departure from the well-known Price tail of linear theory.
  We present an analytical derivation of the power law indices and amplitudes for nonlinear tails stemming from outgoing sources, and validate these predictions to percent-level accuracy with numerical simulations.
  We then perform a perturbative analysis on the dynamical formation of nonlinear tails in a self-interacting scalar field model, wherein the nonlinear tails are sourced by a $\lambda \Phi^3$ cubic coupling.
  Due to cascading mode excitations and back-reactions, nonlinear tails with $t^{-l-1}$ power law are sourced in each harmonic mode, dominating the late time behavior of the scalar perturbations.
  In verification, we conducted numerical simulations of the self-interacting scalar model, including all real spherical harmonic (RSH) with $l\leq 4$ and their respective nonlinear couplings.
  We find general agreement between the predicted and numerical power law indices and amplitudes for the nonlinear tails, with the exception of $l=4$ modes, which display $t^{-4}$ power law instead of the predicted $t^{-5}$.
  This discrepancy is due to distortion in the source of the tails, which is caused by another nonlinear effect.
  These results establish nonlinear tails as universal features of black hole dynamics, with implications for gravitational wave astronomy: they may imprint observable signatures in merger remnants, offering novel probes of strong-field gravity and nonlinear mode couplings.
\end{abstract}
%}

\maketitle
% \flushbottom

\section{Introduction}
\label{sec:introduction}

The advent of gravitational wave (GW) astronomy has ushered in an era of precision tests for strong-field gravity \cite{Berti:2015itd,Berti:2018vdi,Franciolini:2018uyq,LIGOScientific:2020tif}. 
Ground-based detectors like LIGO, VIRGO, and KAGRA, along with the upcoming space-based mission LISA, are reaching sensitivities that enable detailed scrutiny of black hole merger ringdowns \cite{Berti:2005ys,KAGRA:2013rdx,LIGOScientific:2016aoc,KAGRA:2021vkt,LIGOScientific:2023lpe}. %LIGOScientific:2018mvr,LIGOScientific:2020ibl
These observations probe the nonlinear dynamics of spacetime, offering insights into the validity of general relativity and the nature of compact objects. 
A critical aspect of this effort lies in understanding the late-time behavior of perturbations around black holes, where both linear and nonlinear effects leave distinct imprints on the emitted gravitational waves.

Linear black hole perturbation theory (BHPT) has long provided the framework for modeling black hole ringdowns, yielding well-known results such as quasinormal modes (QNMs) and Price’s power law tails \cite{Price:1971fb,Price:1972pw,Leaver:1986gd,Gundlach:1993tp,Gundlach:1993tn,Ching:1995tj,Chandrasekhar:1975zza,Martel:2005ir,Berti:2009kk}.
In order to model ringdown waveforms of astrophysical black holes, numerous efforts had been made toward precision understanding of the quasinormal spectra and Price tails of Schwarzschild and Kerr black holes \cite{Ching:1994bd,Krivan:1996da,Krivan:1997hc,Krivan:1999wh,Burko:2002bt,Burko:2007ju,Hod:2009my,Burko:2010zj,Racz:2011qu,Zenginoglu:2012us,Burko:2013bra,Baibhav:2023clw}.
These linear predictions have been extensively validated through numerical relativity simulations and increasingly precise GW observations.
More recently, nonlinear evolution of black hole perturbations is attracting increasing attention \cite{Sberna:2021eui,Redondo-Yuste:2023seq,Mitman:2022qdl,Cheung:2022rbm,Ioka:2007ak,Redondo-Yuste:2023ipg,Bourg:2024jme,Lagos:2024ekd,Perrone:2023jzq}.
Second-order perturbation theory has revealed phenomena absent in linear treatments, including quadratic quasinormal modes and nonlinear power law tails \cite{Gleiser:1995gx,Gleiser:1998rw,Campanelli:1998jv,Garat:1999vr,Zlochower:2003yh,Brizuela:2006ne,Brizuela:2007zza,Nakano:2007cj,Brizuela:2009qd,Ripley:2020xby,Loutrel:2020wbw,Pazos:2010xf,Khera:2023oyf,Bucciotti:2023ets,Spiers:2023cip,Ma:2024qcv,Zhu:2024rej}.
These findings present exciting opportunities for testing the strong field regime of General Relativity, as well as providing evidence for theories beyond gravity \cite{Berti:2015itd,Berti:2018vdi,Franciolini:2018uyq,LIGOScientific:2020tif,Khera:2024yrk}.

Of particular interest are nonlinear power law tails, which naturally emerge from the outgoing quasinormal profiles present in ringdown dynamics \cite{Okuzumi:2008ej,Lagos:2022otp}.
Recent analytical advances have uncovered the power laws of these nonlinear tails, which are distinct from Price tails and are expected to dominate the late time behavior of black hole ringdowns \cite{Cardoso:2024jme}.
It was found that the power laws of these tails are determined by both the angular number $l$ of the perturbations and the decay exponent $\beta$ of the source.
In addition, recent 3+1 dimensional numerical relativity simulations of black hole mergers have also provided evidence for power law tails distinct from the Price tail \cite{DeAmicis:2024eoy,Ma:2024hzq}.
These developments elucidate the rich phenomenology of nonlinear effects in black hole perturbation theory, calling for new experimental efforts (e.g. GW waveform modeling) to detect their associated signatures.

This paper focuses on the dynamical formation of nonlinear tails; namely, formation of tails due to self-interactions in a scalar field model.
We divide our finding into two main parts.
In the first part, we derive the late time power law indices and amplitudes of nonlinear tails sourced by compact outgoing profiles, and verify them against numerical simulations.
Our analytic results extend and refine earlier works \cite{Okuzumi:2008ej,Lagos:2022otp,Cardoso:2024jme} by providing percent level accurate predictions of tail amplitudes, as well as resolving confusions in the tail power laws.
These improved results allow us to make robust and precise statements on the dynamical nonlinear tails. % in the self-interacting model.
In the second part, we study the dynamical formation of nonlinear tails in a scalar field model with a cubic $\lambda \Phi^3$ coupling via both analytic and numerical means.
Using a perturbative analysis, we argue that nonlinear tails are sourced by nonlinear mode couplings in each harmonic, and provide predictions on their power laws.
In validation, we also perform fully nonlinear numerical simulations that include all harmonics with $l \leq 4$ and their nonlinear mode couplings.
This self-interacting scalar model serves as a prototype for gravitational waves, providing us with insights on nonlinear tails in black hole ringdowns.

We present some highlights of our findings for dynamical nonlinear tail formation.
In our scalar field model, the nonlinear coupling leads to cascading excitation of higher multipoles and back-reactions, generating quasinormal modes that source nonlinear tails.
The perturbative analysis of these processes indicates that all multipoles are eventually dominated by nonlinear tails with $t^{-l-1}$ power laws.
Numerical simulations for the self-interacting scalar model largely confirm predictions in the perturbative analysis, with the exception of $l=4$ multipoles.
For $l \leq 3$ multipoles, we find remarkable accuracy in the predicted power laws and consistency in the amplitudes.
In particular, we observe transition from linear Price tail to nonlinear tail in the initially dominant $(lm)=(11)$ harmonic, which highlights the effect of back-reactions.
For $l=4$ modes, we find $t^{-4}$ power law nonlinear tails instead of the expected $t^{-5}$ tail.
This discrepancy is attributed to nonlinear distortions in the quasinormal waveforms, which in itself lead to open questions on quasinormal waveform modeling.

This paper is organized as follows.
\Sref{sec:nonlinear_tail_analytics} derives the nonlinear tail properties using Green's function methods, comparing our results to prior analytical work.
\Sref{sec:self_interacting_scalar_field} presents analytical analyses and numerical simulations of a fully nonlinear system of self-interacting scalar field, demonstrating how nonlinear tails can be dynamically produced.
We conclude with a discussion of implications for gravitational wave astronomy and open questions in \sref{sec:conclusion}.
Appendices detail the Green's function approximation (\aref{sec:review_of_green_function}), the derivation of sourced tails (\aref{sec:Psi_details}), real spherical harmonic coupling terms (\aref{sec:rsh_couplings}), and the isospectral relationship between Regge-Wheeler and Zerilli equation (\aref{sec:darboux_transform}).

For simplicity, we normalize all dimensionful quantities by the Schwarzschild radius $r_s = 2M$, so that the equations we give consist of only dimensionless quantities.
To convert back to dimensionful quantities, simply multiply the dimensionless quantities by the appropriate $r_s$ factors, e.g. $t \mapsto t r_s$, $r_* \mapsto r_* r_s$ and $\lambda \mapsto \lambda r_s^{-1}$.

\section{Nonlinear tails from Green's function}
\label{sec:nonlinear_tail_analytics}

``Nonlinear tails'' of perturbations around black holes were recently pointed out in ref.~\cite{Okuzumi:2008ej,Lagos:2022otp,Cardoso:2024jme}.
\Revise{Unlike the well known Price tail, nonlinear tails appear in higher order black hole perturbation theory (BHPT). Consider the generic non-linear Regge-Wheeler/Zerilli equation of the form 
\begin{align}
    (\Box + V)\Psi_{l m}  = {\cal O}(\Psi^2, \Psi \partial \Psi, (\partial \Psi)^2)_{lm}, 
\end{align}
the linear order solution $\Psi^{(0)}_{lm}$ generically consist of a part that enters the event horizon and a part that propagates to infinity. The portion entering the horizon is irrelevant for our analysis. By plugging in the linear order outgoing wavepacket in perturbation theory, the higher order wave equation would be of the form 
\begin{align}
      (\Box + V)\Psi_{l m}  = Q(t, r_*)
\end{align}
with 
\begin{align}
   Q(t, r_*)\sim   \frac{F((t-t_i)-(r_*-r_i))}{r^\beta}.
\end{align}
This structure leads to a late-time power-law tail whose index generally differs from the Price tail.
We emphasize that the source term described above is generic within perturbation theory and also naturally includes scenarios involving outgoing matter sources.

In the following, we systematically derive the power law index and amplitude of nonlinear tails sourced by a broad class of outgoing profiles, thereby extending existing results in the literature.}

\subsection{The approximate Green's function}
\label{sec:G_analytical_approximation}
In this subsection, we review and extend some facts on the Green's function of the Regge-Wheeler equation.
We shall focus on the free propagation Green's function $G_F$, which is most relevant to the sourcing of nonlinear tails.

The Regge-Wheeler equation for scalar field with a source term $Q$ is:
\begin{align}
  \label{eq:rw_equation_with_source}
  & (\partial_t^2 - \partial_{r_*}^2 +V(r_*)) \Psi(t,r_*) =  Q(t,r_*) \nonumber \\
  & V(r_*) = l (l+1) \frac{r-1}{r^3} + \frac{r-1}{r^4} \;,
\end{align}
where $r$ is the radial coordinate, $r_* = r + \ln(r - 1)$ is the tortoise coordinate, and $l$ is the azimuthal angular number.
The causal Green's function is defined via:
\begin{align}
  \label{eq:green_function_definition}
  & (\partial_t^2 - \partial_{r_*}^2 + V(r_*)) G(t,r_*;t',r_*') =  \delta(t-t') \delta(r_*-r_*') \qq{such that} \nonumber \\
  & G(t,r_*;t',r_*') = 0 \qq{for} t < t' \;.
\end{align}
The hyperbolic nature of the equation implies causality, that is $G(t,r_*;t',r_*')=0$ for $t-t' < \abs{r_*-r_*'}$.

Conventionally, the Green's function is decomposed into three parts, namely free propagation $G_F$, quasinormal modes $G_Q$, and branch cut $G_B$:
\begin{align} \label{eqn:GreenFQB}
  G(t,r_*;t',r_*') = G_F(t,r_*;t',r_*') + G_Q(t,r_*;t',r_*') + G_B(t,r_*;t',r_*') \;.
\end{align}
These contributions stem from analytic properties of the frequency (Laplace) space Green's function $G(s,r_*,r_*') \equiv \int_{t'}^\infty e^{-s(t-t')} G(t,r_*;t',r_*') \dd{t}$.
More specifically, $G_Q$ arises from poles of $G(s,r_*,r_*')$ on the $\Re[s] < c$ (for some $c >0$) half plane \footnote{The Green's function should be analytic on the $\Re[s] > c$ plane. Usually $c=0$, but there are also cases where poles appear at $\Re[s]>0$~\cite{Hui:2019aox}. There are also Kramers-Kronig relations associated with this analyticity in the right half plane \cite{DeLuca:2024ufn}.}, $G_B$ arises from the branch cut discontinuity of $G(s,r_*,r_*')$ across the negative real $s$ axis, and $G_F$ is due to the pole-like contribution of $G(s,r_*,r_*')$ at $s=0$. 
Note that: since $s=0$ is the branch point for $G(s,r_*,r_*')$, a generic contour around $s=0$ would not form a closed loop on the Riemann surface, and the integral would be sensitive to the contour.
To avoid this ambiguity, we define $G_F$ via:
\begin{align}
  G_F(t,r_*;t',r_*') = \lim_{R \to 0} \frac{1}{2\pi i} \int_{-\pi+0}^{\pi-0} e^{R e^{i\theta} (t-t')} G(R e^{i\theta},r_*,r_*') R e^{i\theta} i \dd{\theta} \;.
\end{align}
As we will see, $G_F$ is responsible for sourcing nonlinear tails.
\begin{figure}
    \centering
    \includegraphics[width=0.5\linewidth]{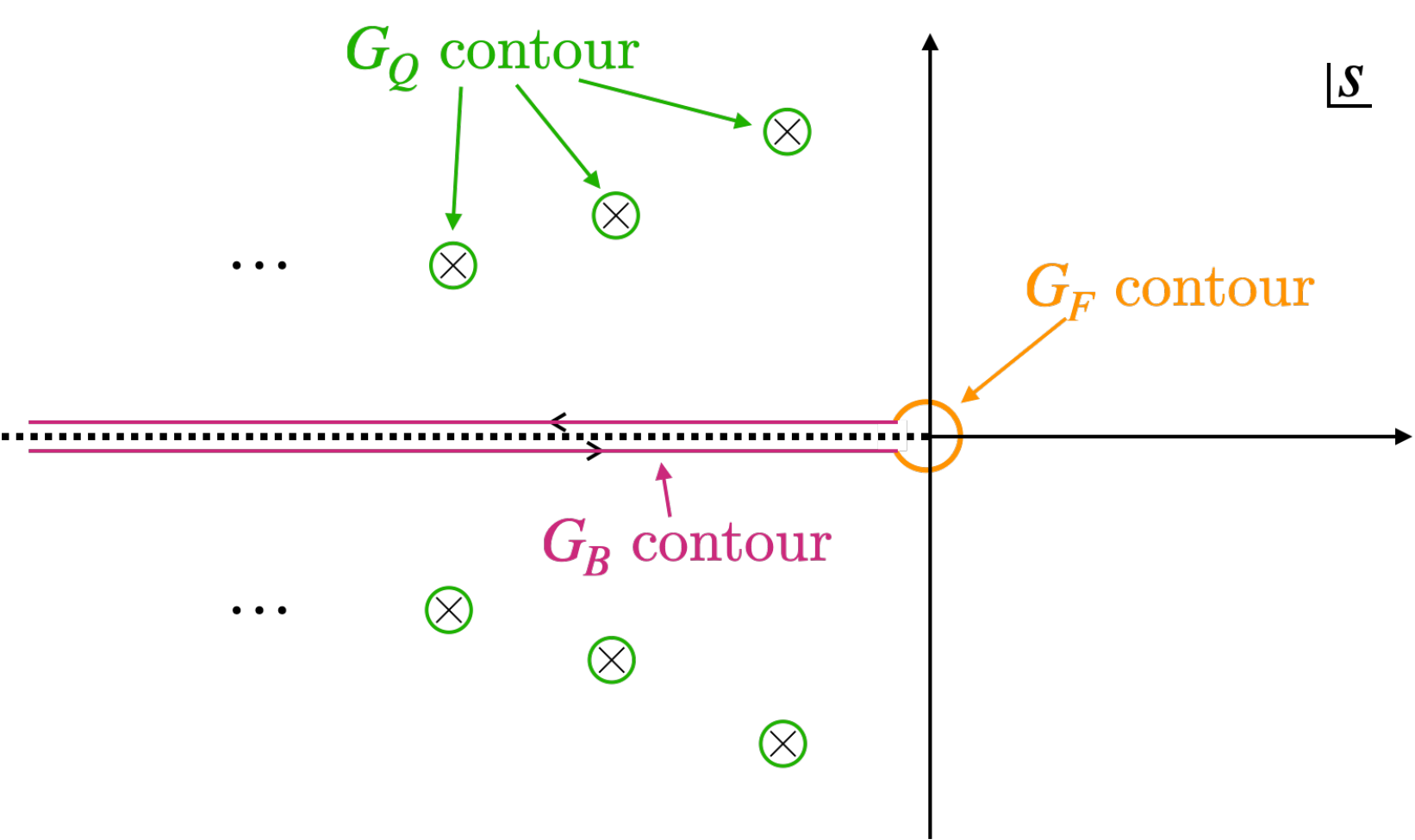}
    \caption{Contours giving rise to corresponding parts of the Green's function in \eqref{eqn:GreenFQB}.}
    \label{fig:Gcontours}
\end{figure}

We first introduce an approximation scheme to the Green's function by exploiting the causal structure of the equation.
Intuitively, a wavepacket propagating in the $r_* \gg 1$ region is only sensitive to the local potential $V(r_*)$.
Since $V(r_*) \approx l(l+1) / r^2$ for $r_* \gg 1$, the evolution of the wavepacket should be approximately governed by equation \eqref{eq:rw_equation_with_source}, but with potential $V$ replaced with $V'(r_*) = l(l+1)/r^2$.
Framed in terms of the Green's function, we expect $G(t,r_*;t',r_*') \approx G'(t,r_*;t',r_*')$, where $G'$ is the Green's function corresponding to the potential $V'$.
To determine the regime of validity of this approximation, we must carefully take into account the causal structure of the equation.
See \fref{fig:green_function_domain_demo} for an illustration of the relevant causal domains.
From the figure, one can tell that the Green's function $G(t,r_*;t',r_*')$ is only sensitive to the value of $V(r_*)$ for $r_* \in [d_1,d_2]$, where $d_1 = (r_*+r_*'-(t-t')) / 2$.
Therefore, as long as $ r_* + r_*' - (t-t') \gg 1$ ($d_1 \gg 1$), $G \approx G'$ is a good approximation.
Physically, the approximation remains valid until the quasinormal mode excited by the source at $(t',r_*')$ emanates from the light ring and propagates to $r_*$.
We present this argument with more mathematical rigor in \aref{sec:causality_arguments}.

\begin{figure}
%  \centering
  \includegraphics[width=0.7\linewidth]{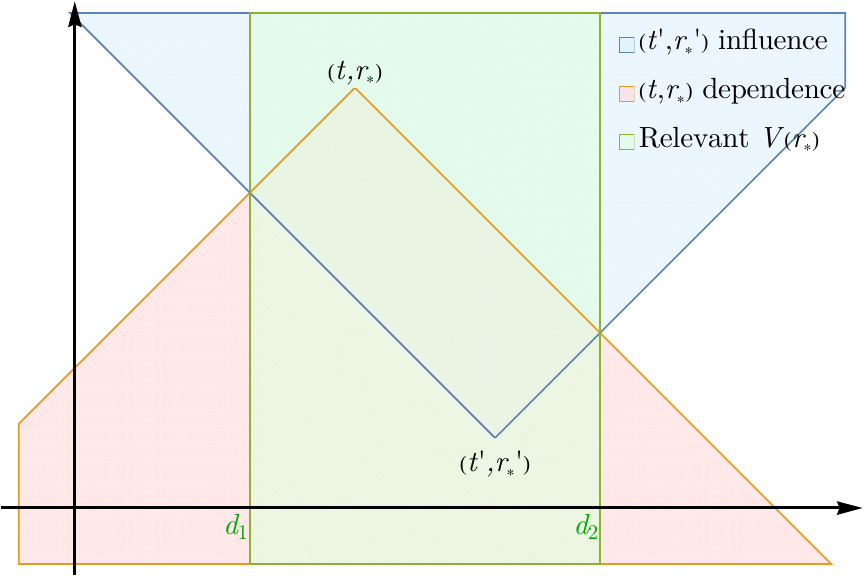}
  \caption{
    The domain of dependence at $(t,r_*)$ is shaded in pink, and the domain of influence from $(t',r_*')$ is shaded in blue.
    For given $(t,r_*)$ and $(t',r_*')$, the Green's function $G(t,r_*;t',r_*')$ is insensitive to the potential outside the green region $\mathbb{R} \times [d_1, d_2]$.
    In the case of Regge-Wheeler potential $V$, we have $V \approx l(l+1) / r^2$ on the region $[d_1,d_2]$ if $d_1 \gg 1$, in which case $G(t,r_*;t',r_*')$ can be approximated by the Green's function corresponding to the $l(l+1)/r^2$ potential.}
  \label{fig:green_function_domain_demo}
\end{figure}

Using the above scheme, we have derived the following approximation for $G_F$ in the case $r_*, r_*' \gg 1$:
\begin{align}
  \label{eq:GF_expressions}
  & \ G(t,r_*;t',r_*') \approx G_F(t,r_*;t',r_*')  \nonumber \\
\approx & \ \frac12 \sum_{n=0}^{l} \frac{(-1)^n (t-t')^{2n}}{2^{l-n} (2n)!}
  \sum_{k=0}^{l-n} r^{2k - l} (r')^{2(l-n-k) - l} \frac{ (2l-1-2k)!! }{k!}
          \frac{ (2n-1+2k)!! }{(l-n-k)!} \nonumber \\
  & \ \textrm{for } \abs{r_* - r_*'} < t-t' < r_* + r_*'\;.
\end{align}
Here, $G \approx G_F$ because the Laplace space Green's function corresponding to $V'(r_*) = l(l+1)/r^2$ have neither quasinormal modes nor branch cut, but only a pole at $s=0$.
The sum expression was derived by explicitly computing the $s=0$ pole contribution.
Details of the derivation appear in \aref{sec:GF_approximation}.
We also performed checks on the approximation by comparing it with a numerically obtained Green's function; see \fref{fig:green_function_demo}.

\begin{figure}
  \centering
  \includegraphics[width=0.7\linewidth]{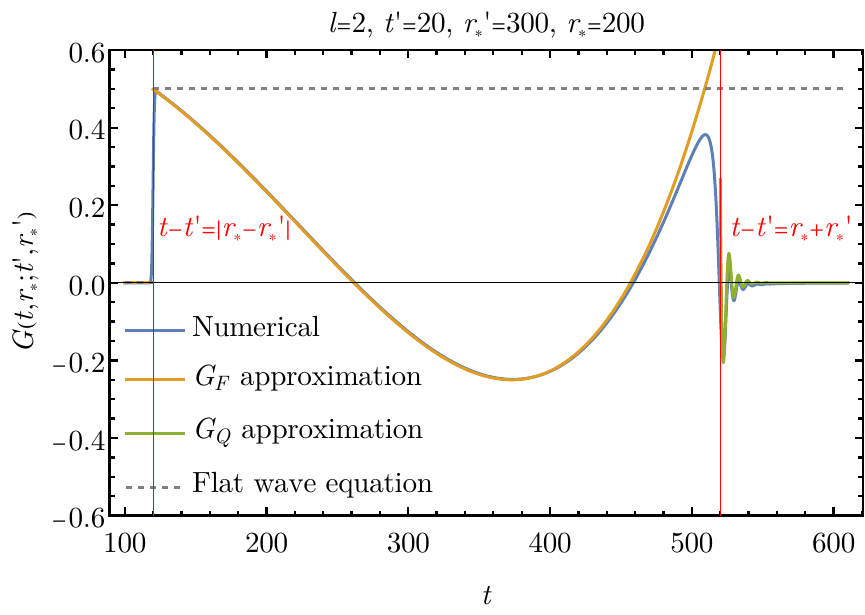}
  \caption{
    Illustration of the $G_F$ approximation \eqref{eq:GF_expressions}.
    For given $l,r_*,r_*',t'$, the $G_F$ approximation \eqref{eq:GF_expressions} (yellow curve) is compared against the full Green's function $G(t,r_*;t',r_*')$ obtained numerically (blue curve).
    The numerical Green's function was obtained by integrating equation \eqref{eq:green_function_definition}, with Dirac deltas replaced by Gaussian distributions with standard deviation $0.5$.
    We see good match between \eqref{eq:GF_expressions} and $G(t,r_*;t',r_*')$ over the range $\abs{r_*-r_*'} < t-t' < r_* + r_*'$.
    The gray dashed line is the Green's function for the flat wave equation, namely $\Theta(t-t'-\abs{r_*-r_*'}) / 2$.
    That the gray line deviates from the blue and yellow curves indicates that $V=0$ is not a sufficiently good approximation.
    The green curve is the predicted quasinormal mode $G_Q$, with quasinormal frequency and excitation factors obtained from table 1 and table 3 of ref.~\cite{Berti:2006wq}; the waveform matches the numerical Green's function.
  }
  \label{fig:green_function_demo}
\end{figure}

Past works often approximate $G_F$ by the flat wave equation Green's function $ \Theta(t-t'-\abs{r_*-r_*'}) / 2$, on the grounds that the potential $V(r_*)$ is negligible for large $r_*$~\cite{Lagos:2022otp,Okuzumi:2008ej}.
However, as is evident from \fref{fig:green_function_demo}, the actual $G_F$ deviates significantly from this flat space approximation.
This deviation occurs because even a small potential $V(r_*)$ can lead to large effects over a long period of time.
In the short time limit, \eqref{eq:GF_expressions} reduces to $G_F \approx \left[1 + \order{(t-t')^2}\right] \times \Theta(t-t'-\abs{r_*-r_*'}) / 2$, which is just the flat space Green's function.
As discussed in ref.~\cite{Cardoso:2024jme}, the flat space approximation is not accurate enough for studying nonlinear tails.
We will show that \eqref{eq:GF_expressions} yields accurate predictions on nonlinear tails.

We further comment on literature works related to approximation \eqref{eq:GF_expressions}.
In ref.~\cite{Barack:1998bw}, the author provides an approximate Green's function in eq.~(28) and (29):
\begin{align}
  \label{eq:Barack_GII_expression}
  G^{I}(u,v; u', v') &= \sum_{n=0}^{l} A_n^l \frac{\frac{\partial^n}{\partial u^n} g_G^{I} (u; u', v')}{(v - u)^{l-n}}, \nonumber \\
  g_G^{I} (u; u', v') &= \frac{1}{l!} \left( \frac{(v' - u)(u - u')}{v' - u'} \right)^l, \quad A_n^l = \frac{(2l - n)!}{n!(l - n)!} \;,
\end{align}
where $u = t - r_*$ and $v = t + r_*$ (similar for $u',v'$) are light cone coordinates.
This Green's function \eqref{eq:Barack_GII_expression} was used to derive nonlinear tails in ref.~\cite{Cardoso:2024jme}.
We found that the sum expression in \eqref{eq:GF_expressions} is the same as $G^{I} / 2$ with $r_*$'s replaced with $r$.\footnote{The $1/2$ factor appears because ref.~\cite{Barack:1998bw} use this alternate convention for the Green's function: $(\partial_u \partial_v + V) G = \delta(u-u') \delta(v-v')$.}
To understand this equality, note that ref.~\cite{Barack:1998bw} approximates the long range potential as $V(r_*) = l(l+1) / r_*^2$, whereas we approximate it as $V(r_*) = l(l+1) / r^2$.
It is thus not surprising that we and the author of \cite{Barack:1998bw} found similar approximations for $G_F$.
Numerically, we found that our approximation \eqref{eq:GF_expressions} is slightly more accurate than \eqref{eq:Barack_GII_expression}.

We will avoid detailed discussion of $G_Q$ and $G_B$ in this work, since they are irrelevant to the nonlinear tail.
The only fact to keep in mind is that the branch cut Green's function $G_B$ sources the Price tail~\cite{Price:1971fb,Leaver:1986gd}, and $G_B \sim t^{-2l-3}$ at late times.
As we will see, depending on the power law of the nonlinear tail, the Price tail could dominate over the nonlinear tail.
A brief review of known results on $G_Q$ and $G_B$ appear in \aref{sec:GQ_and_GB}.

\subsection{Response to an outgoing source}
\label{sec:Psi_analytical_approximation}
To study nonlinear tails, we consider a generic outgoing source:
\begin{align}
  \label{eq:outgoing_source}
  Q(t,r_*) = \frac{F((t-t_i)-(r_*-r_i))}{r^\beta} \Theta(t-t_i) \;,
\end{align}
where $F$ is a function supported on the interval $[-\sigma,\sigma]$.
The source $Q$ describes a profile that switches on at $t_i$ and moves outward at light speed.
We assume $t_i, r_i > 0$, which means the source is active only outside the light ring.

The response due to this source is given by: \footnote{Note that we do not include the homogeneous part of the solution, which is sensitive to the initial conditions \cite{Chavda:2024awq}.}
\begin{align}
  \label{eq:Psi_integral_full}
  \Psi(t,r_*) = \int_{t_i}^t \int_{r_* - (t-t')}^{r_* + (t-t')} [ G_F(t,r_*;t',r_*') + G_Q(t,r_*;t',r_*') + G_B(t,r_*;t',r_*') ] Q(t',r_*') \dd{r_*'} \dd{t'} \;.
\end{align}
 The $G_F$, $G_Q$ and $G_B$ contributions to $\Psi$ will be denoted by $\Psi_F$, $\Psi_Q$ and $\Psi_B$.
This section focuses on $\Psi_F$, where the approximation \eqref{eq:GF_expressions} is valid and $G_F$ is dominant.
See \fref{fig:integration_region} for an illustration on integration domains.
\begin{figure}[t]
  \centering
  \includegraphics[width=0.64\textwidth]{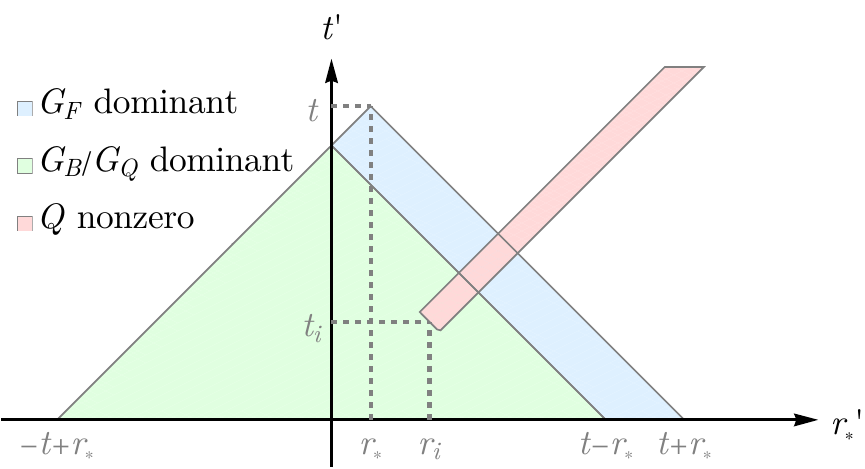}
  \includegraphics[width=0.35\textwidth]{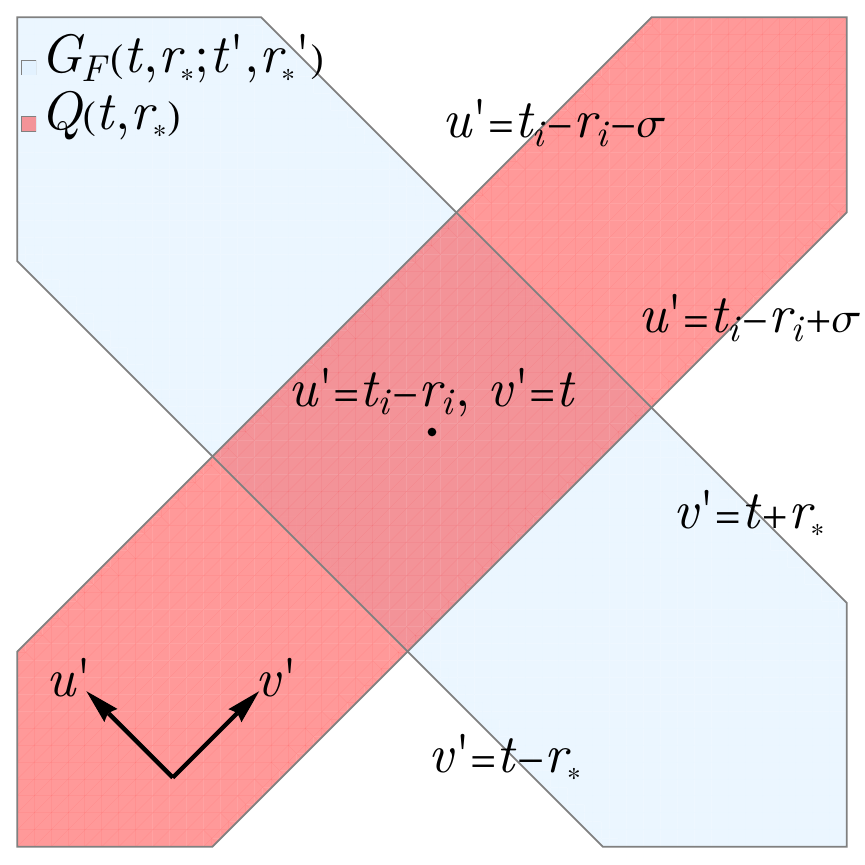}
  \caption{Spacetime regions pertinent to the integral \eqref{eq:Psi_integral_full}.
    The red band illustrates the region where the source \eqref{eq:outgoing_source} is nonzero.
    The blue band illustrates the causal region with $t-r_* \leq v' \leq t+r_*$, where $G_F$ is dominant and $\Psi_F$ is sourced.
  The green region illustrates the causal region where $G_B$ or $G_Q$ are relevant.}
  \label{fig:integration_region}
\end{figure}

\paragraph{$\Psi_F$ derivation}
We first discuss the integration domain relevant for the $G_F$ contribution.
The Green's function $G_F$ is dominant when $-r_* \leq t-t'-r_*' \leq r_*$, and the source $Q$ is nonzero only if $-\sigma \leq (t'-t_i)-(r_*'-r_i) \leq \sigma$.
In terms of the light cone coordinates $u' = t'-r_*'$ and $v' = t' + r_*'$, the domain is the rectangle defined by $t - r_* \leq v' \leq t + r_*$ and $-\sigma + (t_i-r_i) \leq u' \leq \sigma + (t_i-r_i)$.
See \fref{fig:integration_region} for an illustration.
$\Psi_F$ is thus an integral in the light cone coordinates:
\begin{align}
  \Psi_F(t,r_*)
  =  \frac12 \int_{t-r_*}^{t+r_*} \int_{t_i-r_i-\sigma}^{t_i-r_i+\sigma} G_F(t,r_*; t', r_*') Q(t',r_*') \dd{u'} \dd{v'} \;.
\end{align}
The center of the region is $t' = (t+t_i-r_i)/2$, $r_*' = (t+r_i-t_i)/2$, both of which are approximately $t/2$ for $t \gg t_i, r_i$.
In light cone coordinates, the center is $u' = t_i-r_i, v' = t$.

For convenience, we change integration variables from $u'$, $v'$ to $w' \equiv t'-r' - \ln(t/2)$ and $z' \equiv t'+r' + \ln(t/2)$, where $r'$ is the radial coordinate.
By previous discussion on the integration region, in the late time limit ($t \to \infty$) we have $v' \approx t$.
Dropping terms much smaller than $\ln(t)$ (including constant terms such as $r_i,t_i,\sigma,r_*$), we have:
\begin{align}
  & w' \approx u',\quad z' \approx v', \nonumber \\
    & t_i-r_i-\sigma \lesssim w' \lesssim t_i-r_i+\sigma  , \quad
      t-r  \lesssim z' \lesssim t+r  \;.
\end{align}
Thus the integration region in terms of $w'$, $z'$ is approximately rectangular, similar to that for $u'$, $v'$.
The integral is now given by:
\begin{align}
  \label{eq:Psi_F_determinant_factor}
  & \Psi_F(t,r_*)
  \approx  \frac12 \int_{t-r}^{t+r} \int_{t_i-r_i-\sigma}^{t_i-r_i+\sigma} G_F(t,r_*; t', r_*') Q(t',r_*') \det(\pdv{\{u',v'\}}{\{w',z'\}}) \dd{w'} \dd{z'} \nonumber \\
  &\textrm{where } \det(\pdv{\{u',v'\}}{\{w',z'\}}) = \dv{r_*'}{r'} = \left(1 - \frac{1}{r}\right)^{-1} = 1 + \sum_{k=1}^\infty \frac{1}{r^k} \;.
\end{align}
Note that the Jacobian for the coordinate change introduces powers in $1/r$, effectively changing the source to $Q(t',r_*') (1 + 1/r + \order{r^{-2}})$.
We will see that corrections from the Jacobian can sometimes be important for determining the late time behavior of $\Psi_F$.

The power law and amplitude of the nonlinear tail can now be derived.
In the limit $t \gg r_*, r_*', \sigma$ , $G_F / (r')^\beta$ is approximately constant in $w'$ (and $u'$) over the integration region discussed above.
\footnote{This can be seen from \eqref{eq:psi_pm_via_Hankel} and \eqref{eq:G_prime_contour_integral}.  Note that $r'$ and $t'$ dependence in the integrand of \eqref{eq:G_prime_contour_integral} is of the form $e^{-s (t' + r')} p_+((i s r')^{-1}) $, which is roughly constant in $w'$ within the integration region when $t$ is large.  } 
On the other hand, $F((t'-t_i)-(r_*'-r_i))$ is constant in $z'$ (and $v'$).
Thus the integral approximately factorizes into:
\begin{align}
  \label{eq:Psi_F_integral_u_v}
  \Psi_F(t,r_*)
  \approx \ & \sum_{k=0}^\infty I_{l,\beta+k}(t,r_*) \int_{-\sigma}^\sigma F(x) \dd{x} \qq{where} \nonumber \\
  I_{l,\beta}(t,r_*) = \ &  \frac12  \int_{t-r}^{t+r} \frac{G_F(t,r_*;w'=t_i-r_i, z')}{(z'/2)^\beta} \dd{z'} \;.
\end{align}
From above, one can see that $\Psi_F$'s amplitude is proportional to the integral $\int F(x) \dd{x}$, and the exact shape of $F(x)$ does not matter.
Moreover, the $r$ and $t$ power law of $\Psi_F$ depends entirely on the integrals $I_{l,\beta}(t,r_*)$.
The summation in $I_{l,\beta+k}$ comes from the Taylor expansion of the determinant in \eqref{eq:Psi_F_determinant_factor}.
For reference, we list the late time asymptotics of $I_{l,\beta}(t,r_*)$ in \tref{tab:nonlinear_tail_amplitudes}.
One could see that: for fixed $l$, $I_{l,\beta}(t,r_*) \sim r^{l+1} t^{-l-\beta}$ at late times for $\beta \leq 1$ and $\beta \geq l+2$.
On the other hand, $I_{l,\beta}(t,r_*) = 0$ for fixed $l$ and $2 \leq \beta \leq l+1$.
In combination with \eqref{eq:Psi_F_integral_u_v}, we have: for $\beta \leq 1$ or $\beta \geq l+2$, the late time contribution to $\Psi_F$ is dominated by $I_{l,\beta}$ ($k=0$), which means $\Psi_F \sim t^{-l-\beta}$; for $2 \leq \beta \leq l+1$, the dominating contribution to $\Psi_F$ is $I_{l,l+2}$ ($k=l+2-\beta$), which means $\Psi_F \sim t^{-2l-2}$.

The late time $\Psi_F$ power laws largely align with the analytic and numerical results of ref.~\cite{Cardoso:2024jme}.\footnote{See equation (9), (10), (58), (59) and (14) of ref.~\cite{Cardoso:2024jme}.}
Specifically, the authors of ref.~\cite{Cardoso:2024jme} derived a $t^{-l-\beta}$ power law for $\beta \leq 1$ and $\beta \geq l+2$, which matches our results above.
However, they also found a $t^{-l-\beta-1}$ power law for $2 \leq \beta \leq l+1$, whereas we found $\Psi_F \sim t^{-2l-2}$.
There are two reasons for the discrepancy between our result and that of ref.~\cite{Cardoso:2024jme}.
Firstly, the determinant $\dv*{r_*'}{r'}$ coming from the change of variable in \eqref{eq:Psi_F_determinant_factor} was not accounted for in ref.~\cite{Cardoso:2024jme}, so effectively their source neglects additional $1/r^k$ contributions.
Secondly, for $2 \leq \beta \leq l+1$, our calculation gives $I_{l,\beta}(t,r_*) = 0$, whereas ref.~\cite{Cardoso:2024jme} found a $t^{-l-\beta-1}$ power law for the corresponding integral.
Notably, the numerical results in ref.~\cite{Cardoso:2024jme} are consistent with our analytic predictions, including the $t^{-2l-2}$ power laws in the disputed parameter range $2 \leq \beta \leq l+1$.

\renewcommand{\arraystretch}{1.5}
\begin{table}[t]
  \centering
  \begin{tabular}{|c|c|c|c|c|c|c|c|}
    \hline
    $I_{l,\beta}(t,r_*)$ & $\beta=0$ & $\beta=1$ & $\beta=2$ & $\beta=3$ & $\beta=4$ & $\beta=5$ & $\beta=6$ \\
    \hline
    $l=0$&$\frac{r}{2}$&$\frac{r}{t}$&$\frac{2 r}{t^2}$&$\frac{4 r}{t^3}$&$\frac{8 r}{t^4}$&$\frac{16 r}{t^5}$&$\frac{32 r}{t^6}$\\$l=1$&$\frac{r^2}{3 t}$&$\frac{r^2}{3 t^2}$&$0$&$-\frac{4 r^2}{3 t^4}$&$-\frac{16 r^2}{3 t^5}$&$-\frac{16 r^2}{t^6}$&$-\frac{128 r^2}{3 t^7}$\\$l=2$&$\frac{r^3}{5 t^2}$&$\frac{2 r^3}{15 t^3}$&$0$&$0$&$\frac{16 r^3}{15 t^6}$&$\frac{32 r^3}{5 t^7}$&$\frac{128 r^3}{5 t^8}$\\$l=3$&$\frac{4 r^4}{35 t^3}$&$\frac{2 r^4}{35 t^4}$&$0$&$0$&$0$&$-\frac{32 r^4}{35 t^8}$&$-\frac{256 r^4}{35 t^9}$\\$l=4$&$\frac{4 r^5}{63 t^4}$&$\frac{8 r^5}{315 t^5}$&$0$&$0$&$0$&$0$&$\frac{256 r^5}{315 t^{10}}$ \\
    \hline
  \end{tabular}
  \caption{Late time asymptotics of $I_{l,\beta}(t,r_*)$, which determine the spatial and time dependence of the nonlinear tail.
    For non-zero entries, the late time asymptotics follow $\sim r^{l+1} t^{-\beta-l}$.
    For $2 \leq \beta \leq l+1$, we have $I_{l,\beta}(t,r_*) = 0$.
    Also see \eqref{eq:I_integral_examples} for example derivation of these asymptotics.}
  \label{tab:nonlinear_tail_amplitudes}
\end{table}

The nonlinear tails derived in ref.~\cite{Okuzumi:2008ej,Lagos:2022otp} diverge from ours due their use of an approximate free propagation Green's function $G_F$.
In particular, the authors of ref.~\cite{Okuzumi:2008ej,Lagos:2022otp} found a $t^{-2}$ tail for a source decaying like $r_*^{-2}$ (corresponding to $\beta = 2$), which lacks the $l$-dependence seen in our results.
This discrepancy arises because ref.~\cite{Okuzumi:2008ej,Lagos:2022otp} approximate $G_F$ using the flat wave equation Green's function $G_F(t,r_*;t',r_*') \approx \Theta(t - t' - \abs{r_*-r_*'}) / 2$, which is independent of $l$.
As discussed in \sref{sec:G_analytical_approximation} and shown in \fref{fig:green_function_demo}, this flat approximation breaks down when $t-t'$ becomes comparable with $r_*, r_*'$.
Our analysis demonstrates that the form of $G_F$ plays a critical role in determining the power law behavior of nonlinear tails.

\paragraph{$\Psi_Q$ and $\Psi_B$ contribution}
It is well known that $G_Q$ consists of decaying complex exponentials, namely the quasinormal mode spectrum \cite{Leaver:1986gd}.
Since the source $Q$ activates at $(t_i,r_i)$, we expect the earliest quasinormal oscillations to occur around $t \approx t_i + r_i + r_*$.
We claim that, at late times, $\Psi_Q$ is negligible and is dominated by $\Psi_F$ or $\Psi_B$.
While we cannot prove this claim analytically, it is evidenced by our numerical simulations and the analytic calculations in ref.~\cite{Barack:1998bw}.
Additional details appear in \aref{sec:Psi_Q_details}

The branch cut Green's function $G_B$ has power law $G_B \sim t^{-2l-3}$ at late times, yielding a $\Psi_B \sim ^{-2l-3}$ contribution.
This is the known as the Price power law.
For fixed $l$ and $\beta \geq l+4$, $\Psi_B$ dominates over $\Psi_F$ at late times.
% Since we do not have accurate analytic approximations for $G_B$ in the range of interest, we cannot provide accurate predictions on the amplitude of $\Psi_B$.
Additional details appear in \aref{sec:Psi_B_details}.

\paragraph{Summary}
For fixed $l$, the late time power behavior of $\Psi$ is summed up as follows:
\begin{framed}
\begin{itemize}
\item[1.] For $\beta \leq 1$, the late times power law is $t^{-\beta-l}$.
  $\Psi$ is dominated by the sourced tail $\Psi_F$, with its amplitude given by eq.~\eqref{eq:Psi_F_integral_u_v} and \tref{tab:nonlinear_tail_amplitudes}.
\item[2.] For $2 \leq \beta \leq l+2$, the late times power law is $t^{-2l-2}$.
  $\Psi$ is dominated by the sourced tail $\Psi_F$, with its amplitude given by $I_{l,l+2}$ in \tref{tab:nonlinear_tail_amplitudes}.
\item[3.] For $\beta = l+3$, both $\Psi_F$ and $\Psi_B$ yield a $t^{-2l-3}$ power law, so $\Psi \sim t^{-2l-3}$.
\item[4.] For $\beta \geq l+4$, $\Psi_B$ dominates over $\Psi_F$, so $\Psi \sim t^{-2l-3}$.
\end{itemize}
\end{framed}

\subsection{Numerical results}
\label{sec:sourced_tail_numerical}
We numerically solved the sourced Regge-Wheeler equation \eqref{eq:rw_equation_with_source} with the source $Q$ being the compact outgoing profile \eqref{eq:outgoing_source}.
More specifically, given $\beta$, the source $Q$ was \eqref{eq:outgoing_source} with the following parameters:
\begin{align}
  \label{eq:F_x_definition}
  t_i = 10,\quad r_i = 10,\quad
  F(x) = \Revise{ \frac{1}{\sqrt{2\pi} \sigma_F} \exp(\frac{-x^2}{2 \sigma_F^2}) \qq{where} \sigma_F = 0.5} \;.
\end{align}
We take zero initial conditions $\Psi(t=0)=0$ and $\dot{\Psi}(t=0)=0$, and evolved $\Psi$ up to $t=1000$.
The spatial resolution was $h = 0.03$ over domain $r_* \in [-600,1200]$, and 2nd order spatial derivatives were approximated using a 4th order finite difference scheme.
Time steps were fixed to $\Delta{t} = 0.01$, with time evolution performed using an 8th order Runge-Kutta method.
Quadruple precision numbers were used for these simulations.
With these settings, we numerically solved the equations for all combinations $(l,\beta)$, where $0 \leq l \leq 4$ and $0 \leq \beta \leq 6$.
Our code is released at \url{https://github.com/hypermania/BlackholePerturbations}.

\Fref{fig:sourced_Psi_loglog} gives the numerical results for the $\abs{\Psi}$ evolution extracted at $r_*=50$, for $l=1$ and all $\beta$'s.
\footnote{\Revise{The extraction point $r_*$ must satisfy $r_* \gg 1$ for a valid comparison against our analytic results, but otherwise a smaller $r_*$ is preferred to maximize the temporal window for observing $\Psi(t,r_*)$.  We have thus chosen $r_* = 50$. }}
\Revise{Since $F(x)$ \eqref{eq:F_x_definition} is taken to be a Gaussian, equation \eqref{eq:Psi_F_integral_u_v} yields $\Psi_F \approx \sum_{k=0}^{\infty} I_{1,\beta+k}$, which is dominated by $I_{1,\beta}$ (for $\beta \neq 2$) or $I_{1,3}$ (for $\beta = 2$) at late times. For comparison, \fref{fig:sourced_Psi_loglog} also gives $I_{1,\beta}(t,r_*)$ at $r_*=50$ as dashed curves.  }
It is clear that $\Psi(t,r_*)$ exhibit power law tails at late times.
For $\beta \leq 4$, the late time amplitude of $\Psi(t,r_*)$ coincides with prediction to percent level, demonstrating that $\Psi_F$ is indeed the dominant contribution at late times.
This agreement is illustrated by the close resemblance of $\abs{\Psi}$ curves to the $I_{l,\beta}$ lines, which are the predicted late time values of $\abs{\Psi}$.
For $\beta \geq 5$, we see domination by the Price tail at late times.
For $2 \leq \beta \leq 3$, the power law is given by $-2l-2 = -4$, as expected.
We also observe oscillations in the waveform at early times due to the sourced linear quasinormal modes.
The quasinormal oscillations are more pronounced for higher $\beta$'s, possibly because the source is more concentrated in spacetime for these values.
\begin{figure}[t]
  \centering \includegraphics[width=\textwidth]{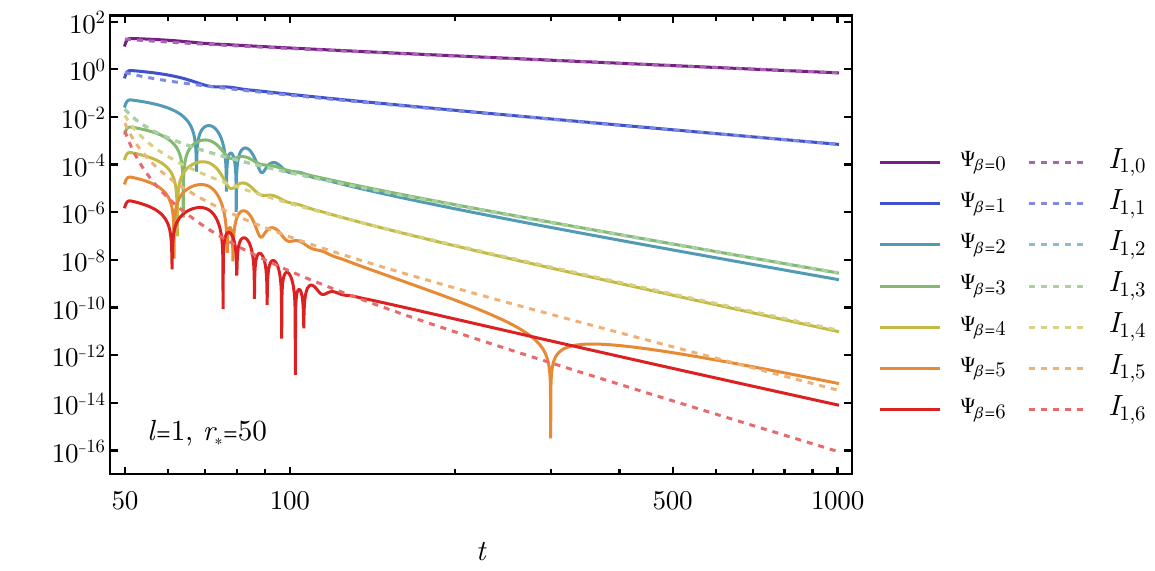}
  \caption{Evolution of $|\Psi(t,r_*)|$ for $l=1$ and $0 \leq \beta \leq 6$.
    The solid curves are from numerical simulation.
    The dashed curves are the predicted late time $|\Psi_F(t,r_*)|$ as listed in \tref{tab:nonlinear_tail_amplitudes}.
    For $\beta \leq 4$, the predicted $|\Psi_F(t,r_*)|$'s and actual $|\Psi(t,r_*)|$ coincide at late times, as expected.
    For $\beta \geq 6$, the Price tail dominates at late times.
    Note that significant early oscillations appear for the higher $\beta$ simulations; they are quasinormal modes excited by the source.}
  \label{fig:sourced_Psi_loglog}
\end{figure}

\Fref{fig:sourced_Psi_powerlaw} gives the evolution of the power law index $\dv*{\ln(|\Psi|)}{\ln(t)}$ for $l=1$.
It is clear that the power laws are converging to values predicted in \sref{sec:Psi_analytical_approximation}.
For $\beta \leq 1$, the power law converges to $p=-l-\beta = -1-\beta$.
For $2 \leq \beta \leq 3$, the power law converges to $p=-2l-2=-4$.
For $\beta = 4,6$, the power law converges to $p=-2l-3=-5$.
For $\beta=5$, the power law had not stabilized at $t=1000$, but one can see a clear trend toward $p=-2l-3=-5$.
A fit (shown in \tref{tab:nonlinear_tail_power_laws}) for the index over time confirms that the index for $\beta=5$ does converge to $-5$ at later times.
\begin{figure}[t]
  \centering \includegraphics[width=0.8\textwidth]{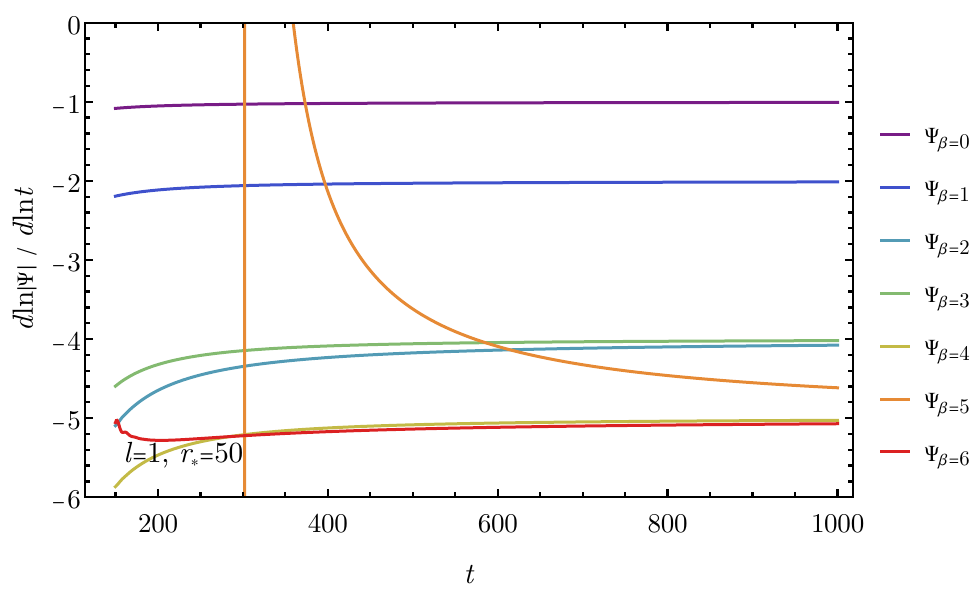}
  \caption{Running of the power law index $\dv*{\ln(|\Psi|)}{\ln(t)}$ for $l=1$ and $0 \leq \beta \leq 6$.
    We get $p=-\beta-l$ for $\beta \leq 1$, $p=-2l-2$ for $\beta \leq l+2$, and $p=-2l-3$ for $\beta \geq l+3$, as expected.
    For $\beta=5$, the power law has not converged to $5$ at the end of simulation $t=1000$, but we expect it to converge at later times.}
  \label{fig:sourced_Psi_powerlaw}
\end{figure}

\Tref{tab:nonlinear_tail_power_laws} gives extrapolated values for the power laws indices.
For some $(l,\beta)$ pairs, the power law had not stabilized at $t=1000$.
We thus extrapolate the power law by fitting the power law index over time $\dv*{\ln(|\Psi|)}{\ln(t)}$ to the template $p(t) = a + b / (t - c)$, and we report the fitted $a$ as the final power law index.
\footnote{This fitting formula is motivated by the fact that when $f = a t^{-n} + ab t^{-n-1}$, $\dv*{\ln(|f|)}{\ln(t)} = -n - b / t + \order{1/t^2}$.}
The fit was performed on index data from range $t \in [700,1000]$, and we have checked that the fits are insensitive to the choice of range.
The power laws largely agree with analytic predictions in \sref{sec:Psi_analytical_approximation}, with the exception of some $2 \leq \beta \leq l+1$ runs exhibiting power law between $t^{-2l-2}$ and $t^{-2l-3}$.
We believe this is due to competition between the $t^{-2l-3}$ Price tail and the $t^{-2l-2}$ nonlinear tail, and will be resolved to pure nonlinear tail dominance in longer simulations.
\begin{table}[t]
  \centering
  \begin{tabular}{|c|c|c|c|c|c|c|c|}
    \hline
    $p$ for $t^{p}$ & $\beta=0$ & $\beta=1$ & $\beta=2$ & $\beta=3$ & $\beta=4$ & $\beta=5$ & $\beta=6$ \\
    \hline
    $l=0$&$0.00$&$-1.00$&$-1.99$&$-2.99$&$-3.03$&$-3.00$&$-2.99$\\$l=1$&$-1.00$&$-2.00$&$-4.00$&$-3.99$&$-4.98$&$-4.99$&$-5.00$\\$l=2$&$-2.00$&$-3.00$&$-6.14$&$-6.99$&$-5.98$&$-6.97$&$-7.00$\\$l=3$&$-3.00$&$-4.00$&$-8.91$&$-8.34$&$-8.01$&$-7.99$&$-8.97$\\$l=4$&$-4.00$&$-5.00$&$-11.00$&$-11.00$&$-10.10$&$-9.99$&$-9.98$ \\
    \hline
  \end{tabular}
  \caption{Late time power laws extrapolated from the sourced simulations.
%    To find these power laws, we fit the instantaneous power law index $\dv*{\ln(|\Psi|)}{\ln(t)}$ using the template $p(t) = a + b / (t - c)$, and report $a$ as the final index.
    For $\beta \leq 1$ and $l + 2 \leq \beta \leq l+3$, we find $t^{-l-\beta}$ power laws, indicating nonlinear tail dominance.
    For $\beta \geq l+4$, we find $t^{-2l-3}$ power laws, indicating Price tail dominance.
    Cases with $2 \leq \beta \leq l+1$ generally exhibit power laws between $t^{-2l-2}$ and $t^{-2l-3}$, signifying competition between the Price tail and the $t^{-2l-2}$ nonlinear tail.
  }
  \label{tab:nonlinear_tail_power_laws}
\end{table}

\begin{table}[t]
  \centering
  \begin{tabular}{|c|c|c|c|c|c|c|c|}
    \hline
    $ \Psi / \Psi_F^{\mathrm{(predicted)}} - 1$ & $\beta=0$ & $\beta=1$ & $\beta=2$ & $\beta=3$ & $\beta=4$ & $\beta=5$ & $\beta=6$ \\
    \hline
    $l=0$ & $0.0015$ & $0.0074$ & $0.0057$ & $-0.12$ & $2.5$ & $180.$ & $8700.$\\
    $l=1$ & $-0.0041$ & $0.00065$ & $-0.47$ & $-0.0014$ & $-0.13$ & $0.92$ & $88.$\\
    $l=2$ & $-0.0091$ & $-0.004$ & $-0.64$ & $-0.96$ & $-0.01$ & $-0.13$ & $0.16$\\
    $l=3$ & $-0.014$ & $-0.0086$ & $0.35$ & $-0.67$ & $-0.57$ & $-0.02$ & $-0.14$\\
    $l=4$ & $-0.019$ & $-0.012$ & $8.2$ & $0.49$ & $-0.46$ & $-0.098$ & $-0.032$ \\
    \hline
  \end{tabular}
  \caption{Fractional error between numerical $\Psi(t)$ and predicted nonlinear tail $\Psi_F^{\mathrm{(predicted)}}(t)$ at $t = 1000$.
    For $\beta \leq 1$ and $\beta = l+2$, the fractional error are all less than $4\%$, indicating an accurate prediction by $\Psi_F^{\mathrm{(predicted)}}$.
    For $\beta = l+3$, the error is $\order{0.1}$, implying that both the nonlinear tail and the Price tail contribute to the $t^{-2l-3}$ tail, with the Price tail being smaller in amplitude.
    For $\beta \geq l+4$, the error range from $0$ to $8700$, indicating dominance by the Price tail.
    Cases with $2 \leq \beta \leq l+1$ generally have $\order{1}$ error, which is likely due to a slow transition from the Price tail to the $t^{-2l-2}$ nonlinear tail.
  }
  \label{tab:nonlinear_tail_amplitude_comparison}
\end{table}

\Tref{tab:nonlinear_tail_amplitude_comparison} provides a comparison between the numerical result and the predicted nonlinear tail at $t=1000$.
We find the fractional error to be $\order{0.01}$ for $\beta \leq 1$ and $\beta = l+2$, and $\order{0.1}$ for $\beta = l+3$, indicating remarkable accuracy for the nonlinear tail predictions.
Note that for $\beta = l+3$, that the fractional error is $\order{0.1}$ suggests the amplitude of the Price tail is $\order{0.1}$ of that of the nonlinear tail.
For $2 \leq \beta \leq l+1$, for which we expect $t^{-2l-2}$ tails, the fractional error is $\order{1}$.
Upon closer inspection, we find that some of the fractional errors for parameters $2 \leq \beta \leq l+1$ have not converged to their asymptotic values at $t=1000$.
Moreover, large fractional errors correlate with the extrapolated power law (see \tref{tab:nonlinear_tail_power_laws}) being closer to $t^{-2l-3}$ rather than the predicted $t^{-2l-2}$.
For example, case $(l,\beta) = (4,2)$ has extrapolated power law index $p=-11=-2l-3$ and fractional error $8.2$, suggesting that $\Psi(t)$ is dominated by the $t^{-2l-3}$ Price tail at $t=1000$.
These numbers are consistent with the picture that there is sensitive competition between the $t^{-2l-3}$ Price tail and the $t^{-2l-2}$ nonlinear tail, and we expect to see only the nonlinear tail in longer simulations.

\section{Dynamical nonlinear tails in a self-interacting scalar field}
\label{sec:self_interacting_scalar_field}
Fields with nonlinear couplings around black holes are a natural setting for the sourced tails discussed in \sref{sec:nonlinear_tail_analytics}.
In this section, we discuss nonlinear tails in scalar perturbations in the case that the scalar field has cubic interactions.\footnote{Although cubic interaction can lead to instabilities, we study only the regime far from the onset of the instability, so our conclusions are also applicable to stable theories with cubic interactions.}
With both analytical perturbative analysis and numerical simulations, we demonstrate that nonlinear tails can form under common initial conditions.
While we studied a scalar field for simplicity, we believe our findings reveal general features in black hole perturbations, including that of gravitational waves.

\subsection{Setup: a scalar field with cubic interaction}
\label{sec:coupled_simulation_setup}

Our setup consists of a real scalar $\Phi$ on a Schwarzschild background with cubic interaction:
\begin{align}
  \label{eq:cubic_scalar_Phi_eqn}
  0 &= -\nabla^a \nabla_a \Phi + \lambda \Phi^2 \nonumber \\
  \dd{s}^2 &= -f \dd{t}^2 + f^{-1} \dd{r}^2 + r^2 \dd{\Omega}^2 \qq{where} f = 1 - \frac{1}{r} \;.
\end{align}
We define $\Psi \equiv r \Phi$, and decompose $\Psi$ in terms of real spherical harmonics (RSH):
\begin{align}
  \Phi(t, r, \hat{r}) &= \frac{1}{r} \sum_{(lm)} \Psi_{lm}(t, r) Y_{lm}(\hat{r}) \;,
\end{align}
where $(lm)$ ranges over the harmonic modes, and $\Psi_{lm}$ and $Y_{lm}$ are real functions.
The mode $\Psi_{lm}$ then satisfies the Regge-Wheeler equation with an additional nonlinear term:
\begin{align}
  \label{eq:cubic_scalar_Psi_lm_eqn}
  & (\partial_t^2 -  \partial_{r_*}^2 + V_l) \Psi_{lm}
    =  \frac{\lambda(1-r)}{r^2} (\Psi^2)_{lm} \nonumber \\
  & V_l(r) =  l (l+1) \frac{r-1}{r^3} + \frac{r-1}{r^4} \;,
\end{align}
where $r_*$ is the tortoise coordinate and $(\Psi^2)_{lm}$ denotes the $(lm)$ component of $\Psi^2$.
More specifically, $(\Psi^2)_{lm}$ is given by:
\begin{align}
  (\Psi^2)_{lm} &\equiv \frac{1}{4\pi} \int Y_{lm} \Psi^2 \dd{\Omega} = \sum_{(l_1 m_1), (l_2 m_2)} C_{lm,l_1 m_1,l_2 m_2} \Psi_{l_1 m_1} \Psi_{l_2 m_2} \nonumber \\
  C_{lm,l_1 m_1,l_2 m_2} &\equiv \frac{1}{4\pi} \int Y_{lm}(\hat{r}) Y_{l_1 m_1}(\hat{r}) Y_{l_2 m_2}(\hat{r}) \dd{\Omega} \;.
\end{align}

The initial condition is an approximately in-going Gaussian packet in the mode $\Psi_{11}$:
\begin{align}
  \label{eq:cubic_scalar_Psi_IC}
  &\Psi_{11}(t=0,r_*) = \frac{1}{\sigma \sqrt{2\pi}} e^{-\frac{(r_*-r_s)^2}{2 \sigma^2 }},\quad
                       \partial_t\Psi_{11}(t=0,r_*) = \partial_{r_*} \Psi_{11}(t=0,r_*) \nonumber \\
  &  \sigma = 0.5,\quad r_s = 50 \nonumber \\
  & \Psi_{lm}(t=0,r_*) = \partial_t \Psi_{lm}(t=0,r_*) = 0 \qq{for} (lm) \neq (11) \;.
\end{align}
With this initial condition, we ran simulations for couplings $\lambda = 0.1, 0.01, 0.001, 0.0001$, up to $t=1200$.

\subsection{Perturbative analysis}
\label{sec:coupled_simulation_analytical}
To understand the effect of the nonlinear term $\lambda \Phi^2$, it is instructive to study the coupled equations \eqref{eq:cubic_scalar_Psi_lm_eqn} in the small $\lambda$ limit.
Given fixed initial conditions, we can write the solution in terms of power series in $\lambda$:
\begin{align}
  & \Psi_{lm}(t, r_*) = \sum_{k=0}^\infty \lambda^k \Psi_{lm}^{(k)}(t, r_*) \;.
\end{align}
If $k$ is the least non-negative integer such that $\Psi_{lm}^{(k)} \neq 0$, we will call $k$ the ``perturbative order'' for $\Psi_{lm}$.
Roughly speaking, a mode $\Psi_{lm}$ of perturbative order $k$ scales as $\lambda^k$ for small $\lambda$.
In particular, if $\Psi_{lm}$ has nonzero initial condition, then its perturbative order is zero, and $\Psi_{lm}^{(0)}$ is simply the solution of the linear Regge-Wheeler equation.

It is customary to solve the system of equations \eqref{eq:cubic_scalar_Psi_lm_eqn} perturbatively.
Expanding the equations order by order in $\lambda$, we have:
\begin{align}
  0 =& \  (-\partial_t^2 +  \partial_{r_*}^2 - V_l) \Psi_{lm}^{(0)} \nonumber \\
        & +  \lambda \left[(-\partial_t^2 +  \partial_{r_*}^2 - V_l) \Psi_{lm}^{(1)} + \frac{(1-r)}{r^2} \left(\Psi^{(0)} \Psi^{(0)}\right)_{lm} \right] \nonumber \\
        & +  \lambda^2 \left[(-\partial_t^2 +  \partial_{r_*}^2 - V_l) \Psi_{lm}^{(2)} + \frac{(1-r)}{r^2} \left(2 \Psi^{(1)} \Psi^{(0)}\right)_{lm} \right] \nonumber \\
        & +  \lambda^3 \left[(-\partial_t^2 +  \partial_{r_*}^2 - V_l) \Psi_{lm}^{(3)} + \frac{(1-r)}{r^2} \left(2 \Psi^{(2)} \Psi^{(0)} + \Psi^{(1)} \Psi^{(1)}\right)_{lm} \right]
          + \hdots \;,
\end{align}
where $\Psi^{(n)} = \sum_{(lm)} \Psi_{lm}^{(n)} Y_{lm}$, and $(\cdot)_{lm}$ denotes the RSH coefficient of mode $(lm)$.
We thus have a sourced Regge-Wheeler equation $(-\partial_t^2 +  \partial_{r_*}^2 - V_l) \Psi_{lm}^{(n)}  = -Q_{lm}^{(n)}$ for each $n$, where the source $Q_{lm}^{(n)}$ is implicitly defined.
It is clear that $Q_{lm}^{(n)}$ contains quadratics of $\Psi_{lm}^{(k)}$, with $k < n$.
Note that the above perturbative scheme readily yields a method for finding the perturbative order of $\Psi_{lm}$: iteratively find $\Psi_{lm}^{(k)}$ for successively higher $k$, and the first $k$ such that $\Psi_{lm}^{(k)} \neq 0$ is the perturbative order.

We now apply the perturbative analysis to initial condition \eqref{eq:cubic_scalar_Psi_IC}.
Since only $\Psi_{11}$ is initially nonzero, it is the only mode of perturbative order zero; namely, $\Psi_{lm}^{(0)} = 0$ for $(lm) \neq (11)$.
Some of the leading order sources $Q_{lm}^{(k)}$ are thus given by
\begin{align}
  Q_{22}^{(1)} = \frac{1-r}{r^2} \left( \frac{1}{2} \sqrt{\frac{3}{5 \pi }} (\Psi^{(0)}_{11})^2 \right),\quad
  Q_{33}^{(2)} = \frac{1-r}{r^2} \left( \frac{3 (\Psi^{(0)}_{11}) (\Psi^{(1)}_{22})}{\sqrt{14 \pi }} \right),\hdots \;.
\end{align}
We provide the full list of leading order sources $Q_{lm}^{(k)}$ for each nonzero harmonic in equation \eqref{eq:Q_lm_leading_order}.
Quadratic terms of form $(\Psi_{11})^2$ source $\Psi_{00}$, $\Psi_{20}$ and $\Psi_{22}$, which are of order $\lambda^1$.
These harmonics in turn source $\Psi_{31}$, $\Psi_{33}$ and induces back-reaction on $\Psi_{11}$ at order $\lambda^2$.
\Tref{tab:Psi_perturbative_orders} lists the harmonics (for $l \leq 4$) that eventually get sourced in the perturbative expansion, along with their perturbative orders.
Harmonics that are not listed do not get sourced at any order, thus they remain exactly zero in the true solution.
For all harmonics except for $(11)$, the perturbative order of $\Psi_{lm}$ is the same as the order of the leading source $Q_{lm}^{(k)}$, since these harmonics have formal solution $\Psi_{lm} = (\partial_t^2 -  \partial_{r_*}^2 + V_l)^{-1} \lambda^k Q_{lm}^{(k)} + \order{\lambda^{k+1}}$.
\begin{table}[h]
  \centering
  \begin{tabular}{|c|c|c|c|c|c|c|c|c|}
    \hline
    $\Psi_{00}$ & $\Psi_{11}$ & $\Psi_{20}$ & $\Psi_{22}$ & $\Psi_{31}$ & $\Psi_{33}$ & $\Psi_{40}$ & $\Psi_{42}$ & $\Psi_{44}$ \\
    \hline
    $\sim\lambda^1$ & $\sim\lambda^0$ & $\sim\lambda^1$ & $\sim\lambda^1$ & $\sim\lambda^2$ & $\sim\lambda^2$ & $\sim\lambda^3$ & $\sim\lambda^3$ & $\sim\lambda^3$ \\
    \hline
  \end{tabular}
  \caption{Perturbative order in $\lambda$ for nonzero $\Psi_{lm}$'s, or the minimum $k$ such that $\Psi_{lm}^{(k)} \neq 0$, for all $l \leq l_{\mathrm{max}} = 4$.
    $\Psi_{lm}$ scales as $\lambda^k$ in the small $\lambda$ limit.
    $\Psi_{11}$ converges to the solution for the linear Regge-Wheeler equation as $\lambda \to 0$.
    Roughly speaking, higher $l$ harmonics are of higher order.
  }
  \label{tab:Psi_perturbative_orders}
\end{table}

Given initial conditions \eqref{eq:cubic_scalar_Psi_IC}, the source terms $Q_{lm}^{(k)}$ \eqref{eq:Q_lm_leading_order} are expected to be approximately in the form \eqref{eq:outgoing_source} with $\beta=1$ at late times.
For example, after the initially ingoing Gaussian wavepacket in $\Psi_{11}$ is reflected at the light ring, it leaves behind a linear quasinormal mode profile, thereby forming a compact outgoing source $\sim (\Psi_{11})^2$.
The QNM profile propagates outwards without decay and is positioned at $r_* \approx t - r_s$ at late times.
Matching \eqref{eq:outgoing_source} with the sources $Q_{lm}^{(k)}$ (e.g. $Q_{00}^{(1)}$) yields an $F$ capturing the quadratic QNM profile (e.g. $F(x) \propto (\Psi_{11}^{(0)}(t,t+r_s-x))^2$), with $t_i-r_i \approx r_s$ and $\beta = 1$.
Harmonics other than $\Psi_{11}$ are also sourced before $t = r_s$ by their (possibly indirect) coupling to $\Psi_{11}$, so we expect similar linear QNM profiles in $\Psi_{lm}$  for $(lm) \neq (11)$.
These QNM profiles in $\Psi_{lm}$ will move along with the QNM profile in $\Psi_{11}$, and they all contribute to $Q_{lm}^{(k)}$ in a similar manner.
That said, as we will see in numerical simulations, there are caveats to the above picture.

As discussed in \sref{sec:nonlinear_tail_analytics}, outgoing compact sources $Q_{lm}^{(k)}$ lead to nonlinear tails that can eventually dominate $\Psi_{lm}$.
Since $\beta = 1$ in our case of interest, the nonlinear tails are expected to exhibit $t^{-l-1}$ power law at late times.
Furthermore, according to the analytic results in \eqref{eq:Psi_F_integral_u_v} and sources \eqref{eq:Q_lm_leading_order}, the amplitudes of the tails depend on the integral
\begin{align}
  \label{eq:approximate_F}
  \int_{-\sigma}^\sigma F(x) \dd{x} \sim \lambda^{k_1 + k_2 + 1} \int_{-\sigma}^\sigma \Psi_{l_1 m_1}^{(k_1)}(t,t+r_s-x) \Psi_{l_2 m_2}^{(k_2)}(t,t+r_s-x) \dd{x}
\end{align}
where $\sigma$ is the rough size of the QNM profile, and $k_1$, $k_2$ are the order of harmonics appearing in the source terms.
The nonlinear tails thus have a predictable amplitude scaling behavior in the small $\lambda$ limit: with the exception of $\Psi_{11}$, all $\Psi_{lm}$'s are expected to have nonlinear tails whose amplitudes scale as given in \tref{tab:Psi_perturbative_orders}.

\subsection{Numerical results}
\label{sec:coupled_simulation_numerical}
We first describe the simulations we ran.
By definition, $(\Psi^2)_{lm}$ contains an infinite sum over quadratic harmonic modes.
We set a cutoff $l_{\mathrm{max}} = 4$ and include all harmonics with $l \leq l_{\mathrm{max}}$ in our simulations.\footnote{To justify the choice to simulate with a cutoff in $l$, note that our perturbative analysis indicates that higher $l$ harmonics typically have stronger $\lambda$ scaling, suggesting that nonlinearities do not lead to large amplitudes in these harmonics. We thus expect harmonics with the lowest $l$'s to capture the most salient features of a full field simulation. In the limit $l_{\mathrm{max}} \to \infty$, we recover a full field simulation using spectral method with RSH basis.}
Each simulation thus consists of $(l_{\mathrm{max}}+1)^2 = 25$ coupled 1+1D PDEs (including harmonics that never get sourced), forming a closed system of equations.
There are $623$ nonzero coupling terms of the form $(\Psi^2)_{lm} \supset \Psi_{l_1 m_1} \Psi_{l_2 m_2}$ for $l_{\mathrm{max}}=4$; all these couplings terms were included.
The spatial resolution was $h = 0.03$ over domain $r_* \in [-600,1200]$, and 2nd order spatial derivatives were approximated using a 4th order finite difference scheme.
Time steps were fixed to $\Delta{t} = 0.01$, with time evolution performed using an 8th order Runge-Kutta method.
Double precision floating point numbers were used for these simulations.
With these settings, we evolved the equations for couplings $\lambda = 0.1, 0.01, 0.001, 0.0001$, from $t=0$ to $t=1200$.
\Revise{We also validated our numerical routine by repeating the $\lambda = 0.001$ simulation with finer discretization ($h=0.015$, $\Delta{t} = 0.005$) or with higher cutoff $l_{\mathrm{max}} = 5$.}
Our code is released at \url{https://github.com/hypermania/BlackholePerturbations}.

\Fref{fig:coupled_Psi_waveform_plot} lists waveforms extracted at $r_*=r_s=50$.
As expected, we observe linear QNM profiles for each $\Psi_{lm}$ after $t = r_* + r_s = 100$.
We extracted the oscillation periods from the waveforms by identifying adjacent peaks, and found that the frequencies are consistent with linear quasinormal frequencies of respective $l$'s.
Upon normalization with respect to $\lambda$, the plotted waveforms for $\lambda=0.1, 0.01, 0.001, 0.0001$ are visually indistinguishable, indicating the accuracy of the $\lambda$ scaling presented in \tref{tab:Psi_perturbative_orders}.
Note that for $\Psi_{40}$, $\Psi_{42}$ and $\Psi_{44}$, the QNM oscillations are clearly displaced from zero.
This displacement is due to sourcing by the QNM profiles at lower perturbative orders, with the response characterized primarily by the free propagation Green's function $G_F$.
We will discuss the implications of this displacement in detail later.
%Notably, we also see low frequency ringing before $t=100$, which can be attributed to the dispersiveness of free propagation Green's function $G_F$.

\begin{figure}[t]
  \centering \includegraphics[width=\textwidth]{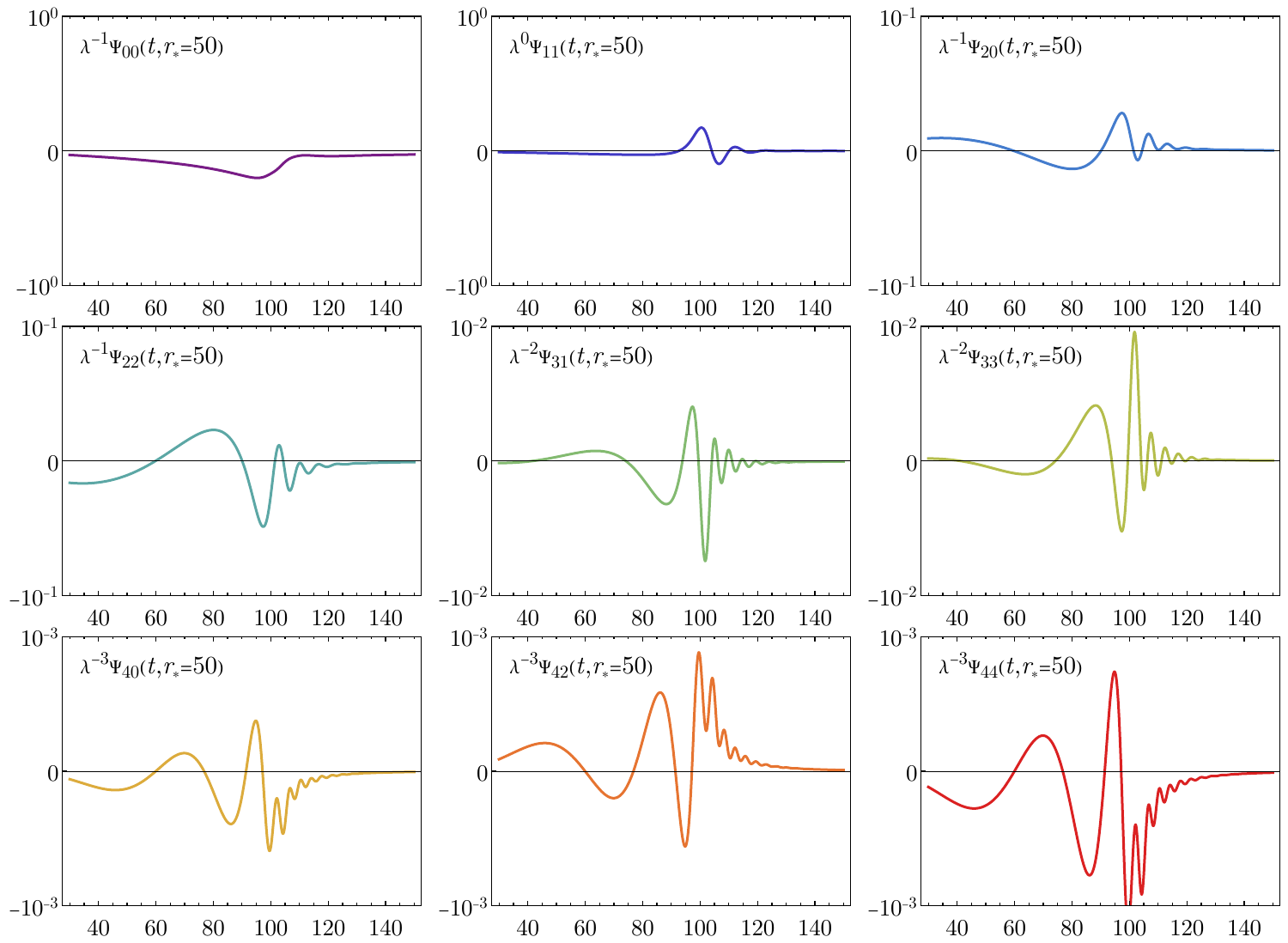}
  \caption{Waveforms for the activated modes.
    The waveform is extracted at $r_*=50$ round $t=100$, when the quasinormal modes from the source first appear.
    These waves are outgoing, and acts as a compact source on other modes due to the cubic nonlinear coupling.
    Although this figure is generated from the simulation with $\lambda = 0.1$, the corresponding (scaled) waveforms for $\lambda = 0.01, 0.001, 0.0001$ are visually indistinguishable from above.
    This shows $\Psi_{lm}$ scales with $\lambda$ as expected, indicating the validity of the perturbative analysis.}
  \label{fig:coupled_Psi_waveform_plot}
\end{figure}

\Fref{fig:coupled_Psi_evolution_plot} illustrate the late time evolution of $\Psi_{lm}$ for $\lambda = 0.001$, $\lambda = 0.01$ and $\lambda = 0.1$.
We observe clear indications of nonlinear power law tails.
In both simulations, $\Psi_{00}$, $\Psi_{20}$, $\Psi_{22}$, $\Psi_{31}$, $\Psi_{33}$ all roughly follow $t^{-l-1}$ power laws at late times, consistent with analytical predictions on the nonlinear power law tails.
(The $l=4$ harmonics do not follow $t^{-l-1}$ power, which will be discussed later this subsection.)
A closer inspection on the instantaneous power laws $\dv*{\ln(|\Psi_{lm}|)}{\ln(t)}$ shows that they are closer to $-l-1$ in the $\lambda = 0.001$ simulation than for the $\lambda = 0.1$ coupling, signifying distortion of the power laws due to nonlinearities in the $\lambda = 0.1$ case.
Within the same simulation, tails with the same angular number $l$ differ only by an $\order{1}$ factor, which stems from the differences between the RSH couplings \eqref{eq:Q_lm_leading_order}.
A comparison between amplitudes of the tails in the $\lambda = 0.001$ and $\lambda = 0.1$ simulations shows that the nonlinear tails also follow the expected $\lambda$ scaling.
\footnote{Similar scaling was discussed in \cite{Zlochower:2003yh}.}

% Check amplitude prediction
We also perform a consistency check on the amplitudes of the nonlinear tails.
According to equation \eqref{eq:Psi_F_integral_u_v} and our perturbative analysis, at late times, the predicted amplitude of the nonlinear tail is given by
\begin{align}
  \Psi_{lm}^{(\mathrm{predicted})}(t) \approx I_{l,1}(t,r_*) \times \int_{-\sigma}^\sigma F_{lm}(x) \dd{x} \;,
\end{align}
where $F_{lm}$ is inferred from the source \eqref{eq:Q_lm_leading_order} and has generic form \eqref{eq:approximate_F}.
Because we don't have accurate analytical predictions for the QNM profiles around the wavefront, we chose to approximate $F_{lm}$ by directly extracting the QNM profiles from our simulation results.
For example, for $(lm)=(22)$, we have
\begin{align}
  F_{22}(x)
%  \approx \lambda \left(-\frac{1}{2}\sqrt{\frac{3}{5\pi}} \right)  \left( \Psi_{11}^{(0)}(t,t+r_s-x) \right)^2
  \approx \lambda \left(-\frac{1}{2}\sqrt{\frac{3}{5\pi}} \right) \left( \Psi_{11}(t_0,t_0+r_s-x) \right)^2 \;,
\end{align}
where $t_0$ is some time chosen after the QNM profile has stabilized, and we have approximated $\Psi_{11}^{(0)} \approx \Psi_{11}$.
%$F_{22}$ is expected to be insensitive to the exact choice of $t_0$.
Other $F_{lm}$'s can be found via a similar prescription.
Using the numerical result from the $\lambda=0.001$ run, we evaluate the integral on $F_{lm}(x)$ at $t_0=100$ over the range $-50 < x < 50$ ($\sigma = 50$).
The predicted amplitudes $\Psi_{lm}^{(\mathrm{predicted})}$ at $t=1200$ are then compared with the numerical results $\Psi_{lm}$:
\begin{align}
  \label{eq:coupled_tail_amplitude_comparison}
  \frac{\Psi_{lm}(t)}{\Psi_{lm}^{(\mathrm{predicted})}(t)} \approx
  \begin{cases}
    1.07 & (lm) = (00) \\
    1.15 & (lm) = (20), (22) \\
    1.39 & (lm) = (31), (33) \\
    24.8 & (lm) = (40), (42), (44)
  \end{cases}
  \;.
\end{align}
We found good accuracy in the predicted amplitudes for $l \leq 3$, with better accuracy in lower $l$'s.
For $l=4$, the actual amplitude is much larger than predicted; the cause of this discrepancy is discussed in more detail later this subsection.

The behavior of mode $\Psi_{11}$ exhibits competition between the nonlinear tail and the Price tail.
One can see from \fref{fig:coupled_Psi_evolution_plot} that, at late times, $\Psi_{11} \sim t^{-2l-3} = t^{-5}$ for $\lambda = 0.001$ and $\Psi_{11} \sim t^{-l-1} = t^{-2}$ for $\lambda = 0.1$.
The reason for this distinction is that the nonlinear tail is weaker in amplitude by a factor of $(0.001 / 0.1)^2$ in the $\lambda = 0.001$ simulation than that for the $\lambda = 0.1$ simulation.
Throughout the $\lambda = 0.001$ simulation, the Price tail ($\sim t^{-5}$) dominates, whereas the nonlinear tail ($\sim t^{-2}$) dominates in the $\lambda = 0.1$ simulation.
More visibly, in the $\lambda = 0.01$ simulation, a clear transition from Price tail to nonlinear tail occurs around $t=550$; the sharp dip at the transition is the result of the two power law tails possessing different signs.
Matching the amplitudes for the two contributions ($t^{-5} \sim \lambda^2 t^{-2}$) tells us that the time $t_{\mathrm{transition}}$ of transition from Price tail to nonlinear tail scales as $t_{\mathrm{transition}} \propto \lambda^{-2/3}$, and could be far beyond the span of simulation for small $\lambda$.

\begin{figure}[t]
  \centering
  \includegraphics[width=0.49\textwidth]{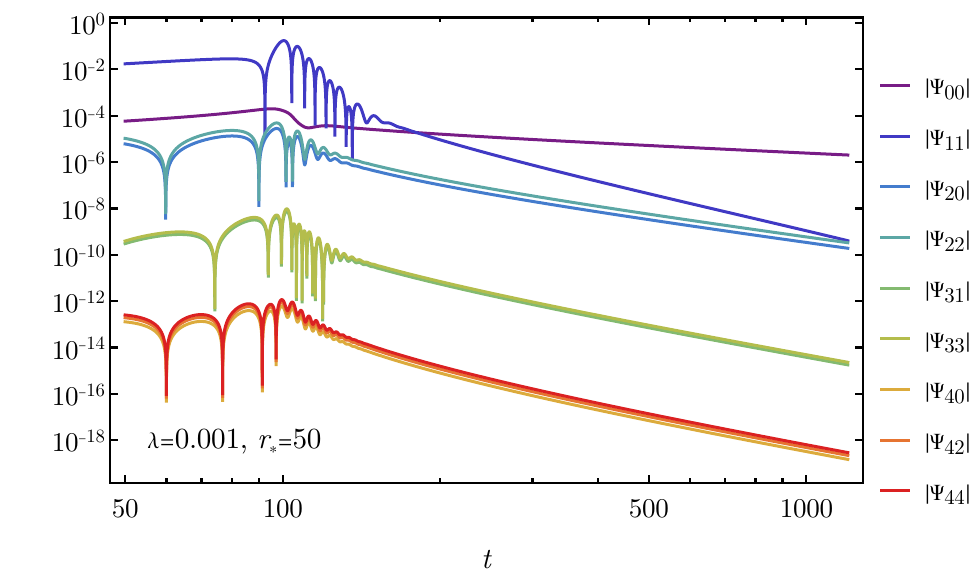}
  \includegraphics[width=0.49\textwidth]{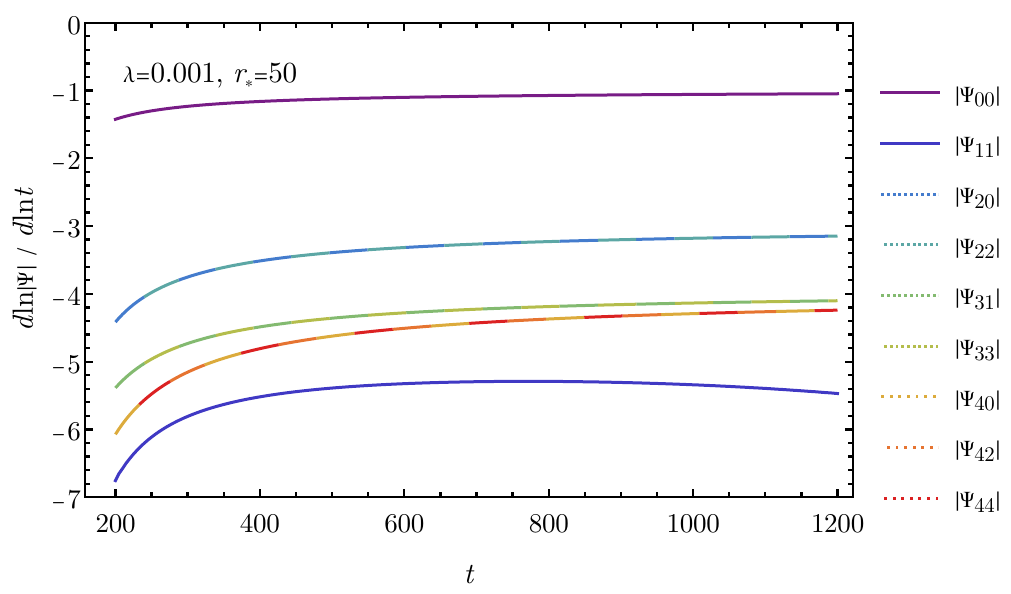}\\
  \includegraphics[width=0.49\textwidth]{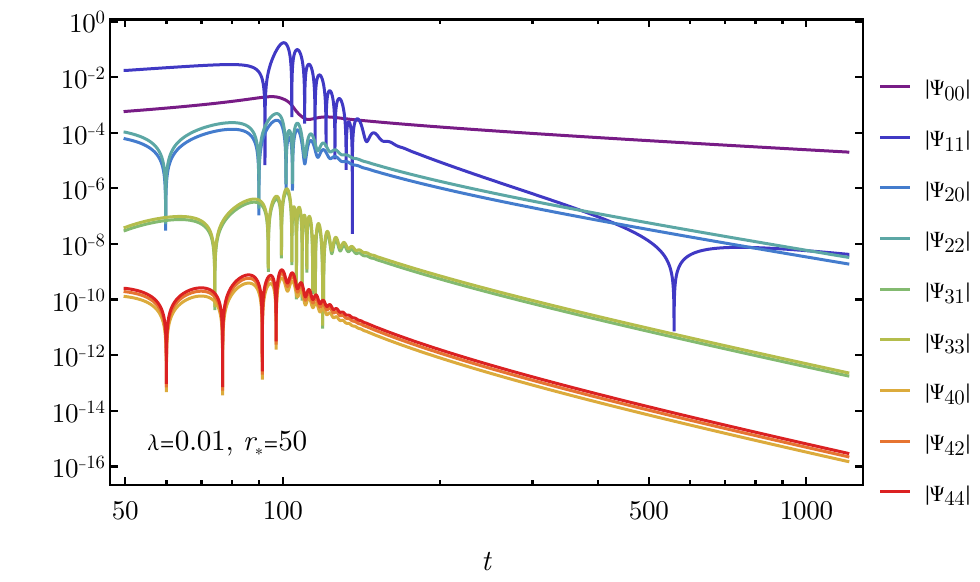}
  \includegraphics[width=0.49\textwidth]{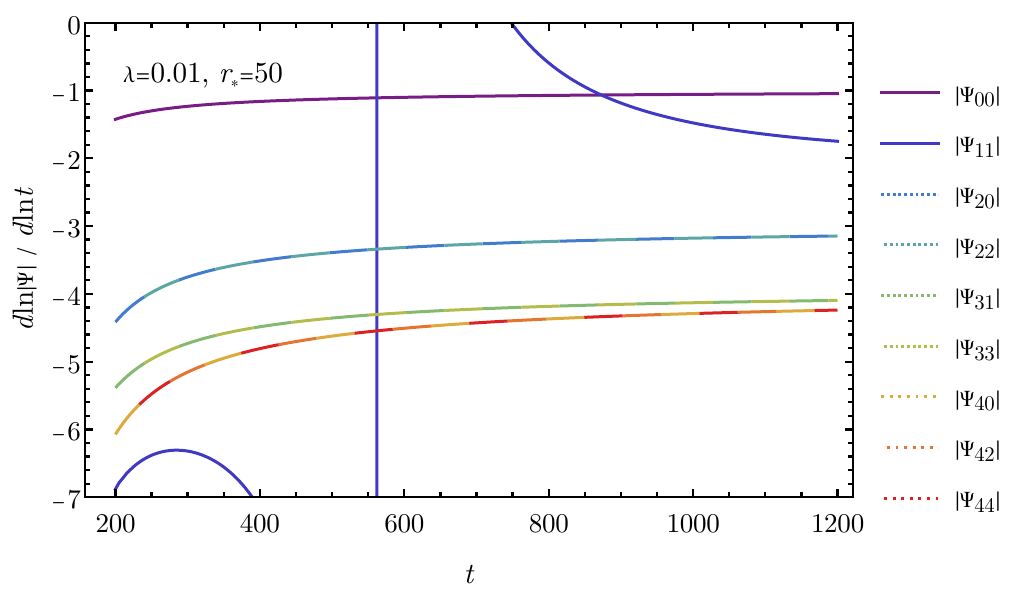}\\
  \includegraphics[width=0.49\textwidth]{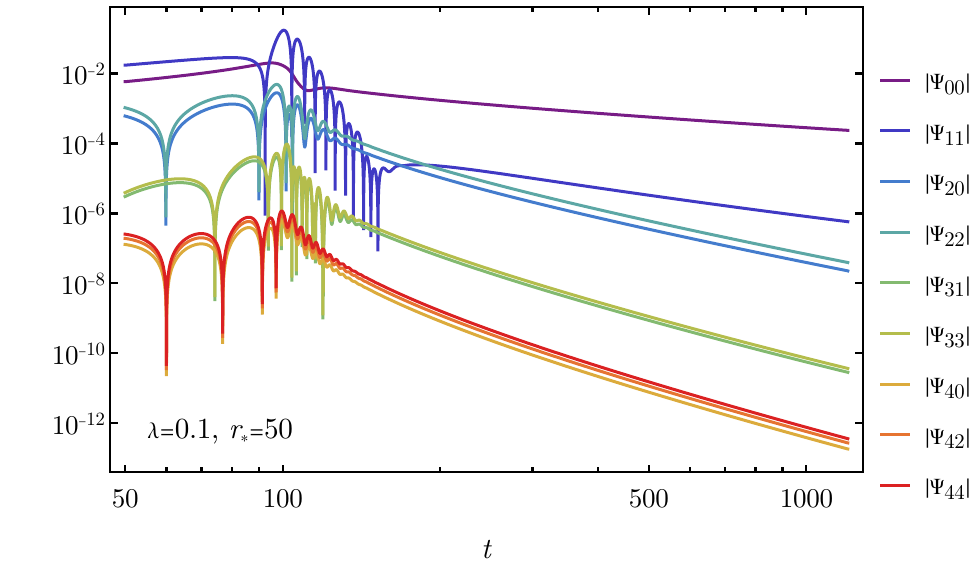}
  \includegraphics[width=0.49\textwidth]{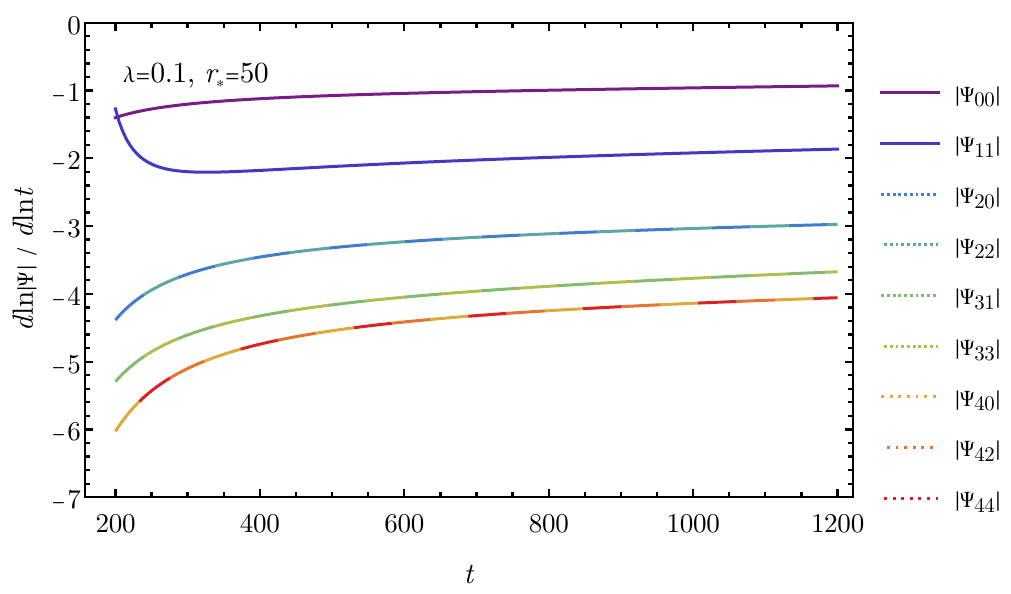}
  \caption{Amplitudes and instantaneous power law indices for the activated modes.
    Most modes are found to follow expected $t^{-l-1}$ nonlinear tail power law.
    As $\lambda$ is varied from $0.001$ to $0.1$, we see expected amplitude scaling in the excited harmonics, and competition between nonlinear tail and Price tail in $\Psi_{11}$.}
  \label{fig:coupled_Psi_evolution_plot}
\end{figure}

The behavior of the $l=4$ harmonics is worth special attention.
By extrapolating the instantaneous power laws $\dv*{\ln(|\Psi_{lm}|)}{\ln(t)}$ in \fref{fig:coupled_Psi_evolution_plot}, we find that $\Psi_{40}$, $\Psi_{42}$ and $\Psi_{44}$ exhibit $t^{-4}$ tails at late times.
This result deviates from our naive expectations that the nonlinear tails satisfy $t^{-l-1} = t^{-5}$ power laws.
This deviation is due to waveform distortions in the lower-$l$ harmonics that source $l=4$ harmonics.
\Fref{fig:coupled_Psi_waveform_evolution_plot} contain snapshots of outgoing profile in $\Psi_{11}$ and $\Psi_{33}$ at different $t$.
One can see that, while the $\Psi_{11}$ waveform remains roughly fixed, the $\Psi_{33}$ waveform clearly has a component (in addition to the QNM profile) that increases in time.
In fact, we find similar additional components in the $l=2$ and $l=4$ waveforms, with the component more noticeable (compared to the QNM profile) for higher $l$.
The source $\Psi_{11} \Psi_{33}$ is thus not merely propagating outward, but also growing in amplitude, effectively changing the $\beta = 1$ index in \eqref{eq:outgoing_source} to $\beta=0$.
\Revise{Note that the $t^{-4}$ power law behavior in the $l=4$ tails, and the distortions in the $\Psi_{33}$ waveform, are all leading order effects that persist in the $\lambda \to 0$ limit.}

The distortions in the outgoing waveforms are themselves due to nonlinear interactions.
For example, the distortion in $\Psi_{33}$ is nonlinearly sourced by the outgoing $\Psi_{11} \Psi_{22}$ profile.
To understand the origin of this distortion, we consider a simplified calculation, wherein the source profile (e.g. $\Psi_{11} \Psi_{22}$) is approximated by a Dirac delta, and $V_l \approx 0$:
\begin{align}
  & (\partial_t^2 -  \partial_{r_*}^2) \Psi =  \delta(t-r_*),\quad
    \Psi(t=0) = \dot{\Psi}(t=0) = 0 \nonumber \\
    % & \Psi(t,r_* = t+d) = \Theta(-d) \Theta(d+2t) \frac{t+d}{2} \nonumber \\
    % & \Psi(t,r_* = t+d) =
    % \begin{cases}
    %   0 & d < -2t \\
    %   (2t+d) / 4 & -2t < d < 0 \\
    %   0 & d > 0
    % \end{cases}
  & \Psi(t,r_* = t-u) =
    \begin{cases}
      0 & u > 2t \\
      (2t-u) / 4 & 0 < u < 2t \\
      0 & u < 0
    \end{cases}
    \;.
\end{align}
Here, $f(u) = \Psi(t, r_* = t-u)$ is the response waveform extracted at time $t$.
We can see that, the response $\Psi$ is a linear function lagging behind the source, and its amplitude is growing as $t/2$ at $u = 0$ (at the outgoing profile).
This response waveform is roughly consistent with the time-dependent component of the waveform in $\Psi_{33}$ (see \fref{fig:coupled_Psi_waveform_evolution_plot}), which is growing over time and dominating over the QNM profile.
The $t^{-4}$ power law in the $l=4$ tail should thus be understood as a ``higher order'' nonlinear effect: it is the nonlinear effect caused by another nonlinear effect.
To corroborate these claims, we have checked and found similar distortions in the sourced simulations of \sref{sec:nonlinear_tail_analytics}, which capture the contribution from the source (e.g. $\Psi_{11}\Psi_{22}$) at late times.
%A more accurate calculation of the response waveform can be done using \eqref{eq:GF_expressions}.
%Note that this is a different calculation from the nonlinear tail, since for the tail we fix $r_*$ and study the late time behavior, whereas here $r_*=t+d$ is set to be moving along with the outgoing profile.

\begin{figure}[t]
  \centering
  \includegraphics[width=0.44\textwidth]{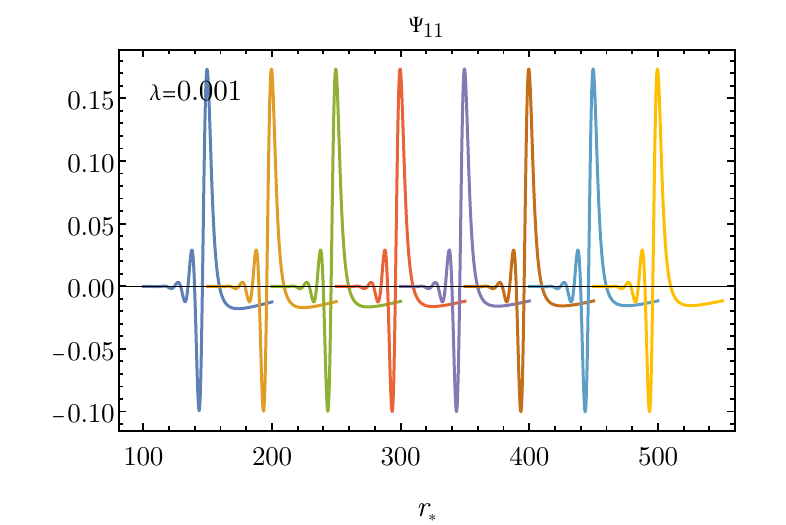}
  \includegraphics[width=0.55\textwidth]{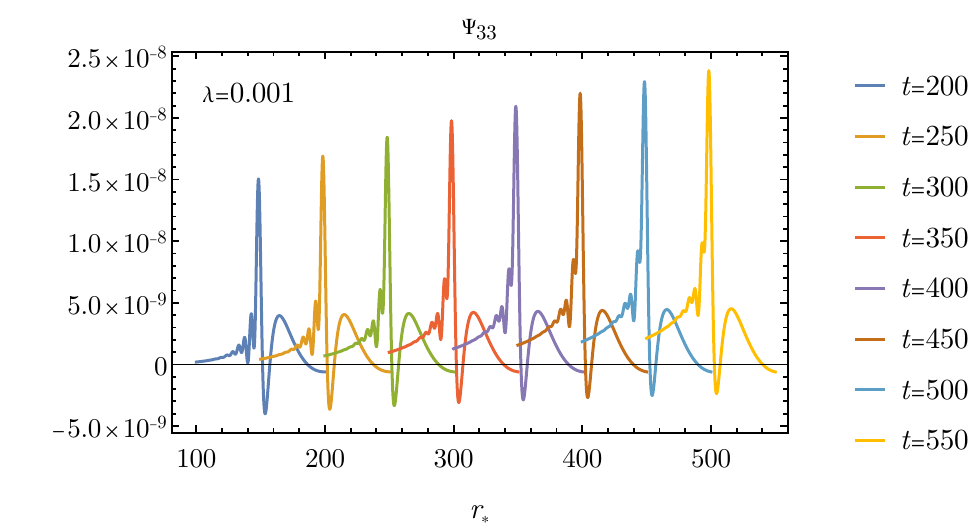}
  \caption{Snapshots of the waveforms in $\Psi_{11}$ and $\Psi_{33}$ as they propagate to spatial infinity.
    The $\Psi_{11}$ waveform is largely fixed over time, but significant distortion accumulate in the $\Psi_{33}$ waveform due to nonlinear effects.
    Such distortions are sourced by couplings of the form $\sim \Psi_{11} \Psi_{22}$.
    This effect is also confirmed by the sourced simulations in \sref{sec:sourced_tail_numerical}.
  }
  \label{fig:coupled_Psi_waveform_evolution_plot}
\end{figure}

\section{Discussion and conclusion}
\label{sec:conclusion}

Using both analytical and numerical methods, we systematically analyzed the scalar nonlinear tails sourced by compact outgoing profiles.
We have shown analytically that the free propagation component \( G_F \) \eqref{eq:GF_expressions} of the Green's function, derived through an approximation scheme accounting for causal domains, determines the nonlinear tail structure.
This clarifies the conceptual origin of the nonlinear tails to be distinct from the branch cut Green's function $G_B$.
Furthermore, we analytically derived both the power law and the amplitude of the nonlinear tails due to outgoing sources decaying as $r^{-\beta}$, and found good matches with numerical results for $0 \leq l \leq 4$ and $0 \leq \beta \leq 6$.
These analytic results, summarized in eq.~\eqref{eq:Psi_F_integral_u_v} and \tref{tab:nonlinear_tail_amplitudes}, refine prior analyses that revealed the power law but not the amplitude of the nonlinear tails \cite{Cardoso:2024jme}.
That the predicted amplitudes match the numerical ones to percent level validates our analytic understanding of nonlinear tails independently from the power law.
In addendum, we found a $t^{-2l-2}$ power law for $2 \leq \beta \leq l+1$ due to higher-order terms in the Jacobian of tortoise coordinate transformations (eq.~\eqref{eq:Psi_F_determinant_factor}), a phenomenon previously discovered in numerical simulations but not fully understood \cite{Cardoso:2024jme}.

We performed a perturbative analysis on scalar perturbations around a Schwarzschild spacetime, with spherical harmonic mode couplings induced from cubic self interaction (eq.~\eqref{eq:cubic_scalar_Psi_lm_eqn}).
For an initial condition wherein only the $(11)$ harmonic is nonzero, the self-interacting scalar develops cascading mode excitations, with each harmonic's amplitude scaling as some power in the coupling $\lambda$.
Typically, higher multipoles scale more intensely in $\lambda$, suggesting that the nonlinear effects are predominantly captured by the lowest multipoles.
Furthermore, utilizing the sourced tail analyses in \sref{sec:Psi_analytical_approximation}, we argue analytically that the outgoing linear quasinormal profiles are expected to source nonlinear tails with $t^{-l-1}$ power laws.
In particular, we provide a link (eq.~\eqref{eq:approximate_F}) between the quasinormal profile and the amplitude of corresponding nonlinear tails.

Numerical simulations of a cubic scalar field model confirm the analytical predictions while revealing additional nonlinear phenomena.
While lower multipoles ($l \leq 3$) exhibit clean $t^{-l-1}$ tails consistent with first-order nonlinear effects, higher multipoles display modified power laws due to secondary nonlinear processes (\fref{fig:coupled_Psi_evolution_plot}).
Specifically, $l=4$ modes develop $t^{-4}$ tails rather than the predicted $t^{-5}$, arising from waveform distortions in their source terms (\fref{fig:coupled_Psi_waveform_evolution_plot}).
In addition, tail amplitudes for the lower multipoles ($l \leq 3$) are consistent with the extracted outgoing quasinormal profiles, whereas tail amplitudes for $l=4$ exhibit large discrepancy from the predicted values (eq.~\eqref{eq:coupled_tail_amplitude_comparison}).
The discrepancies are again due to nonlinear effects, wherein quasinormal waveforms are distorted by nonlinear sourcing from lower multipoles, leading to altered source profile for the nonlinear tails.
We emphasize that both the nonlinear tail and the distortion in its source terms originate from the free propagation Green's function $G_F$, highlighting its importance in understanding nonlinear effects in black hole perturbations.

Our analytical and numerical analyses indicate that, in the presence of cubic scalar interactions, nonlinear tails eventually dominate each harmonic.
Due to the cubic coupling, nonlinear tails are sourced in each harmonic, including those from lower multipoles or via back-reaction.
These nonlinear tails exhibit $t^{-l-1}$ power laws and dominate over the $t^{-2l-3}$ Price tails as well as the exponentially decaying quasinormal modes at late times.
In particular, we observe clear transition from the linear Price tail to the nonlinear tail in the $(11)$ harmonic (\fref{fig:coupled_Psi_evolution_plot}).
Since the amplitude of the nonlinear tail (but not the linear Price tail) scales with coupling $\lambda$, we derived a $t_{\mathrm{transition}} \propto \lambda^{-2/3}$ scaling for the time at which the $(11)$ harmonic transitions to nonlinear tail domination.
For coupling $\lambda = 0.1$, the nonlinear tails dominates at $t_{\mathrm{transition}} \lesssim 200$, whereas for $\lambda = 0.001$ the domination happens for $t_{\mathrm{transition}} > 1000$.

% Implications for GW detection
The cubic interaction model serves as a prototype for more complex nonlinear effects in full general relativity, where the gravitational waves (GW) from black hole merger remnants may yield similar tail signatures.
Specifically, the outgoing quasinormal profiles from remnant black holes satisfy Regge-Wheeler or Zerilli equations, which admit the $V \approx l(l+1)/r^2$ approximation in this work and are thus amenable to our analytic analyses.
\footnote{Also see \aref{sec:darboux_transform} for the isospectral connection between these two equations.}
Assuming that the GW source terms decay as $r^{-2}$ and applying the results in \sref{sec:Psi_analytical_approximation}, we expect to see $t^{-2l-2}$ power law nonlinear tails that dominate the GW waveform at late times.
In fact, a $t^{-2l-2}$ tail ($l=4$ in their case) is already seen in a second order simulation of metric perturbations in ref.~\cite{Cardoso:2024jme}.
Evidence for such tails are also recently seen in 3+1 dimensional numerical relativity (NR) simulations of merging black holes \cite{DeAmicis:2024eoy,Ma:2024hzq}, but whether these tails are of the nonlinear nature discussed in this work remains to be answered.

Precision predictions of nonlinear tail amplitudes in GWs can pave way for novel probes of gravity.
One direction is to explore the consistency between the amplitude of the nonlinear tail and its source.
In eq.~\eqref{eq:coupled_tail_amplitude_comparison}, we performed consistency checks between the quasinormal mode profiles and the amplitudes of the corresponding sourced tails.
With upcoming GW detectors, it may become feasible to extract both the quasinormal modes and the nonlinear tails from GW signals of merger events, opening up the possibility for a similar test.
Analogous to consistency relationships between linear and quadratic quasinormal modes~\cite{Khera:2023oyf,Khera:2024yrk,Cheung:2022rbm,Mitman:2022qdl,Ma:2024qcv,Perrone:2023jzq}, the consistency between quasinormal modes and nonlinear tails can be a test of General Relativity.
\footnote{This consistency check can readily be done on NR simulations (e.g. ref.~\cite{DeAmicis:2024eoy,Ma:2024hzq}) to understand the nature of the observed power law tails.}
Another possible avenue is to search for a transition from Price tail to nonlinear tail in GW waveforms.
For nonlinear tails in GW, linear Price tail and nonlinear tail from back-reaction become comparable when $t^{-2l-3} \sim \lambda^2 t^{-2l-2}$, namely $t_{\mathrm{transition}} \propto \lambda^{-2}$.
The onset timing of the transition is thus a sensitive probe of nonlinear couplings in metric perturbations, and can potentially be used to constrain theories of modified gravity~\cite{Berti:2015itd}.
Even if the nonlinear tails are generically too strong for a transition to occur (as in our $\lambda = 0.1$ simulation), we might still find transitions in scenarios where the Price tail is enhanced compared to the nonlinear tails, such as extreme mass ratio inspirals~\cite{Islam:2024vro,DeAmicis:2024not}.

Our findings suggest multiple directions for further investigation.
First, the observed waveform distortions in higher multipoles warrant systematic study of nonlinear propagation effects.
While our analysis focus on nonlinear tails, one can see from \fref{fig:coupled_Psi_waveform_evolution_plot} that nonlinear effects can be important for waveform modeling of higher multipole quasinormal modes.
Second, extending our perturbative framework to gravitational perturbations, with interactions derived from Einstein Hilbert action, can reveal how nonlinear tails manifest in gravitational wave waveforms \cite{Brizuela:2009qd}.
Finally, to understand nonlinear tail phenomena in broader astrophysical context, it is imperative to generalize our analysis to Kerr spacetimes.

Nonlinear tails represent a fundamental feature of black hole perturbation theory, bridging linear quasinormal mode descriptions and fully nonlinear dynamics.
Our results demonstrate that these tails are inherent properties of wave propagation in curved spacetime.
As gravitational wave detectors approach design sensitivity \cite{LIGOScientific:2018mvr,LIGOScientific:2020ibl,KAGRA:2021vkt}, understanding such nonlinear effects will be crucial for interpreting late-time signals and probing the nonlinear regime of strong-field gravity.

\acknowledgments

We thank Valerio De Luca, Lam Hui, Justin Khoury and Huan Yang for useful discussion. This project is supported by APRC-CityU New Research Initiatives/Infrastructure Support from Central.

\appendix
\section{Approximations of Green's function}
\label{sec:review_of_green_function}

In this section, we review and extend approximations to the Green's function of the Regge-Wheeler equation \eqref{eq:green_function_definition}.
In \aref{sec:causality_arguments} and \aref{sec:GF_approximation}, we derive approximation \eqref{eq:GF_expressions} of the free propagation Green's function $G_F$.
In \aref{sec:GQ_and_GB}, we review literature works on quasinormal mode ($G_Q$) and branch cut ($G_B$) contributions to the Green's function.
To contextualize these approximations, we summarize the dominant contributions across spacetime regions:
\begin{align}
  G(t,r_*;t',r_*') \approx
  \begin{cases}
    0 & t-t' < \abs{r_* - r_*'} \\
    G_F & \abs{r_* - r_*'} < t-t' < \abs{r_*} + \abs{r_*'} \\
    G_Q + G_B & t-t' \gg \abs{r_*} + \abs{r_*'}
  \end{cases}
  \;.
\end{align}

\subsection{Green's function's causal dependence on the potential}
\label{sec:causality_arguments}
As discussed in \sref{sec:G_analytical_approximation}, it is intuitive that the value of potential $V(r_*)$ far away from $r_*$ and $r_*'$ cannot affect the Green's function $G(t,r_*;t',r_*')$.
In this section, we formalize this intuition and prove that $G(t,r_*;t',r_*')$ only depends on $V(r_*)$ restricted to $[d_1,d_2]$, where $d_1 \equiv \frac{r_*+r_*' - (t-t')}{2}$ and $d_2 \equiv \frac{r_*+r_*' + (t-t')}{2}$.
See \fref{fig:green_function_domain_demo} for an illustration of the relevant domains.
This argument adds to existing efforts on understanding the Green's function by exploiting causal structure~\cite{Hui:2019aox,Lagos:2022otp}.

\begin{theorem}
  Fix $r_*, r_*'$ and $t-t' \geq 0$.
  $G(t,r_*;t',r_*')$ depends only on the potential $V$ restricted to domain $D \equiv [d_1, d_2]$, where $d_1 \equiv \frac{r_*+r_*' - (t-t')}{2}$ and $d_2 \equiv \frac{r_*+r_*' + (t-t')}{2}$.
  In other words, let $V'$ be a potential such that $\eval{V'}_D = \eval{V}_D$, and let $G'(t,r_*;t',r_*')$ be the Green's function satisfying $(\partial_t^2 - \partial_{r_*}^2 + V'(r_*)) G'(t,r_*;t',r_*') =  \delta(t-t') \delta(r_*-r_*')$, then $G(t,r_*;t',r_*') = G'(t,r_*;t',r_*')$ for the given $t,r_*,t',r_*'$.
\end{theorem}
\begin{proof}
  Let $\Delta(t,r_*;t',r_*') \equiv G'(t,r_*;t',r_*') - G(t,r_*;t',r_*')$, then:
  \begin{align}
   (\partial_t^2 - \partial_{r_*}^2 + V'(r_*)) \Delta(t,r_*;t',r_*') = (V(r_*)  - V'(r_*)) G(t,r_*;t',r_*') \;.
  \end{align}
  We have $\Delta = 0$ if the source $ (V(r_*)  - V'(r_*)) G(t,r_*;t',r_*')$ is zero on the domain of dependence of $(t,r_*)$.
  Let $I = \{(\tau,\rho) | t' \leq \tau \qq{and} \abs{r_*' - \rho} \leq \tau - t'\}$ be the domain of influence starting from $(t',r_*')$, then the source is zero for $(t,r_*) \notin I$.
  Also, the source is zero on $(t,r_*) \in \mathbb{R} \times D$ since $\eval{V'}_D = \eval{V}_D$.
  Thus $\Delta(\tau,\rho;t',r_*') = 0$ if the domain of dependence at $(\tau,\rho)$ have null intersection with the region
  \begin{align}
    R = \mathbb{R}^2 - ((\mathbb{R}^2 - I) \cup (\mathbb{R} \times D))
    = I \cap (\mathbb{R}^2 - (\mathbb{R} \times D))
    = I \cap (\mathbb{R} \times (\mathbb{R} - D)) \;.
  \end{align}
  Recall that the domain of dependence at $(\tau,\rho)$ is $S = \{(\tau',\rho') | \tau' \leq \tau \qq{and} \abs{\rho' - \rho} \leq \tau - \tau'\}$.
  For $S \cap R = \emptyset$ to be satisfied, a sufficient condition is  $\tau + \rho \leq  (d_2 - r_* + t') + d_2 = t + r_*$ and $\tau - \rho \leq (-d_1 + r_*' + t')- d_1 = t - r_*$.
  In particular, $\tau = t$ and $\rho = r_*$ satisfies this condition, so $G(t,r_*;t',r_*') = G'(t,r_*;t',r_*')$.
\end{proof}

\subsection{Approximation of $G_F$}
\label{sec:GF_approximation}
In this section, we derive approximation \eqref{eq:GF_expressions} and clarify its regime of validity.
%valid for $\abs{r_*-r_*'} < t-t' < r_* + r_*'$.

Following the arguments in \aref{sec:causality_arguments}, we consider the following Green's function $G'$ that is approximately $G$ far from the horizon ($r_*, r_*' \gg 1$):
\begin{align}
  \left( \partial_t^2 - \partial_r^2 + \frac{l(l+1)}{r^2} \right) G'(t,r_*;t',r_*') =  \delta(t-t') \delta(r-r')  \;,
\end{align}
where we have used $V(r_*) \approx l(l+1)/r^2$, $\partial_{r_*}^2 \approx \partial_r^2$ and $\delta(r_*-r_*') \approx \delta(r-r')$.
To derive an expression for $G'$, we employ frequency (Laplace) domain method, and write down an explicit form of corresponding the frequency space Green's function $g'(r_*,r_*',s)$~\cite{Leaver:1986gd}:
\begin{align}\label{eq:G_prime_frequency_domain}
  &G'(t,r_*;t',r_*') = \frac{1}{2\pi i} \int_{c - i\infty}^{c + i\infty} g'(r_*, r_*',s) e^{s \Delta t} \dd{s},\quad  \Delta t \equiv t - t' \nonumber \\
  &g'(r_*, r'_*, s) = \frac{\psi_{\infty_-}(r_{*<}, s) \psi_{\infty_+}(r_{*>}, s)}{W(s)}, \quad    r_{*>} \equiv \max(r_*, r_*'), \quad r_{*<} \equiv \min(r_*, r'_*) \;.
\end{align}
Here, \( W(s) \equiv \psi_{\infty_+} \psi_{\infty_-,r} - \psi_{\infty_-} \psi_{\infty_+,r}\) is the Wronskian constructed from $\psi_{\infty_-}$ and $\psi_{\infty_+}$, the two WKB solutions of equation
\begin{equation}
  \partial_{r}^2 \psi - \left(s^2 + \frac{l(l+1)}{r^2}\right) \psi = 0 \;.
\end{equation}
It is easy to verify that this equation has solutions in terms Hankel functions, with the following form:
\footnote{It is possible to obtain these Hankel function solutions by approximating from exact series solutions of the frequency domain Regge-Wheeler equation. 
  Ref.~\cite{Leaver:1986gd} by Leaver provides these solutions in eq.~(18):
  \begin{align}
    \psi_{\infty_\pm}(r, s) &= (2is)^{\pm s} e^{\pm i \phi_\pm} (1 - 1/r)^s \sum_{L=-\infty}^{\infty} b_L [G_{L+\nu}(\eta,z) \pm i F_{L+\nu}(\eta,z)] \;,
  \end{align}
  where $\eta = -is$ and $z = i s r$.
  Since we are interested in the $s = 0$ pole contribution, we take the $s\to 0$ limit, which yields $\eta \to 0$, $(2is)^s \to 1$, $(1-1/r)^s \to 1$, $b_L \to 0$ for $L \neq 0$, and $\nu \to l$~\cite{Leaver:1986gd,Leaver:1986vnb}.
  These relations imply:
  \begin{align}
    \psi_{\infty_\pm}(r, s) \sim F_{l}(0,z) \mp i G_{l}(0,z) = \sqrt{\frac{\pi}{2} z} H_{l+1/2}^{(1,2)}(z)
    \qq{(up to a phase)} \;.
  \end{align}
}
\begin{align}
  \label{eq:psi_pm_via_Hankel}
  \psi_{\infty_+}(r_*, s)
  &= \sqrt{\frac{\pi}{2} z} H_\nu^{(1)}(z)
  = p_+(z^{-1}) e^{-sr} \nonumber \\
  \psi_{\infty_-}(r_*, s)
  &= \sqrt{\frac{\pi}{2} z} H_\nu^{(2)}(z)
  =  p_-(z^{-1}) e^{sr} \;,
  % = p_\pm(z^{-1}) e^{iz} + q_\pm(z^{-1}) e^{-iz}
\end{align}
where $\nu = l+1/2$, $z=isr$, and $p_\pm$ are (finite order) polynomials with $\abs{p_\pm(0)} = 1$.
\footnote{To see that Hankel functions with half integer indices can be written in the above form, see eq.~(10.1.8) and (10.1.9) of ref.~\cite{abramowitz+stegun}.}
One can see that $\psi_{\infty_+}$ and $\psi_{\infty_-}$ are exactly the outgoing and ingoing solution at $r_* \to \infty$, in accordance with the notations of Leaver~\cite{Leaver:1986gd}.
Moreover, the $G'$ integral \eqref{eq:G_prime_frequency_domain} can be seen to contain exponential factors like $e^{s(\Delta{t} - \abs{r_* - r_*'})}$, which means its integration contour should be closed to the $\Re[s] < 0$ half plane only if $\Delta{t} > \abs{r_* - r_*'}$, as required by causality.

It is immediate from \eqref{eq:psi_pm_via_Hankel}  that the Wronskian $W(s)=2s$ and that $\psi_{\infty_\pm}(r_*, s)$ is holomorphic on $\mathbb{C}$ except for a pole at $s=0$.
The absence of a branch cut in the solutions \(\psi_{\infty_\pm}\) implies that 
\(g(r_*,r_*',s)\) does not have a branch cut on the $s$-plane either.
Therefore, the only contributing part of the $s$-integral in eq.~(\ref{eq:G_prime_frequency_domain}) will come from the pole at $s=0$:
\begin{align}
  \label{eq:G_prime_contour_integral}
  G'(r_*, t; r_*', t')  &= \frac{1}{2 \pi i} \int_{C} e^{s \Delta{t}} \frac{\psi_{\infty_-}(r_*, s) \psi_{\infty_+}(r_*', s)}{2s} \dd{s} \;,
\end{align}
where $C$ is a circle contour around $s=0$, and we have assumed $r_* < r_*'$ for concreteness. The integrand has a pole at $s=0$. 
%By using the formal definition of the \textit{residue theorem} which is given as
%\[
%\int_C f(s) \, ds = 2\pi i \sum_{i=1}^{n} \text{Res}(f, s_i),
%\]
%where \( \text{Res}(f, s_i) \) is the residue of \( f(s) \) at \( s = s_i \). In our case, \( s^k \) has a simple pole at \( s = 0 \). Thus,
%\[
%\text{Res}(s^k, s = 0) = \lim_{s \to 0} s^k = \delta(k, 0),
%\]
%which is nonzero only when \( k = 0 \).
To get the residue at $s=0$, we can expand $e^{s(t-t')}$ and $\psi_{\infty_\pm}$ in powers of $s$, and extract the coefficient for $s^{-1}$.
Concretely, we have:\footnote{See eq.~(10.1.2) and (10.1.3) of ref.~\cite{abramowitz+stegun}.}
\begin{align}
    e^{s \Delta{t}} &= \sum_{k=0}^\infty \frac{(s \Delta{t})^k}{k!} \nonumber \\
  \psi_{\infty_\pm}(r, s) &= P(r, s) \mp i Q(r, s) \qq{where} \nonumber \\
   P(r, s) &= z^{l+1} \sum_{k=0}^\infty \frac{(-z^2/2)^k }{k! (2l+2k+1)!!},\quad Q(r, s) = z^{-l} \sum_{k=0}^{\infty} \frac{(z^2/2)^k  (2l-1-2k)!! }{k!} \;.
\end{align}
Collecting the integrand as a power series in $s$, we have
\begin{align}
  e^{s\Delta{t}} \psi_{\infty_-}(r, s) \psi_{\infty_+}(r', s) &= \sum_{k=k_{\mathrm{min}}}^\infty A_k(r, r', \Delta{t}) s^k
\end{align}
and
\begin{equation}
     G'(r_*, t; r_*', t') = \frac{1}{2}\sum_{k=0}^{\infty} A_k(r, r', \Delta{t}) \delta(k, 0)=\frac{1}{2}A_0(r, r', \Delta{t}) \;.
\end{equation}
Due to the $z^{l+1}$ and $z^{-l}$ prefactors in \( P(r, s) \) and \( Q(r, s) \), only the product \( Q(r, s) Q(r', s) \) can contribute to the $s^0$ term \( A_0(r, r', \Delta{t}) \).
Now,
\begin{align}
  &\left(\sum_{k=0}^\infty \frac{s^k(\Delta{t})^k}{k!}\right) Q(r,s)Q(r',s) \nonumber \\
  =& \sum_{k_1=0}^\infty \frac{s^{k_1}(\Delta{t})^{k_1}}{k_1!}
     z^{-l} \sum_{k_2=0}^{\infty} \frac{(z^2/2)^{k_2}  (2l-1-2k_2)!! }{k_2!}
     (z')^{-l} \sum_{k_3=0}^{\infty} \frac{((z')^2/2)^{k_3}  (2l-1-2k_3)!! }{k_3!} \nonumber \\
  =& \sum_{k_1=0}^\infty \frac{(\Delta{t})^{k_1}}{k_1!} (-1)^{l} r^{-l} (r')^{-l} \times \nonumber \\
  & \sum_{k=0}^\infty (-1/2)^k s^{k_1 + 2 k - 2 l} \sum_{k_2=0}^k r^{2k_2} (r')^{2(k-k_2)} \frac{ (2l-1-2k_2)!! }{k_2!} \frac{ (2l-1-2(k-k_2))!! }{(k-k_2)!} \;.
\end{align}
The \( s^0 \) term arises from the condition \( k_1 + 2k - 2l = 0 \), or equivalently when \( k = l - \frac{k_1}{2} \).
The nonzero terms only occur when \( k_1 = 0, 2, \dots, 2l \).
Renaming \( n = \frac{k_1}{2} \) and $k_2$ to $k$, the \( s^0 \) term is:
\begin{align}
  \sum_{n=0}^{l} \frac{(\Delta{t})^{2n}}{(2n)!} (-1)^n (1/2)^{l-n}
  \sum_{k=0}^{l-n} r^{2k - l} (r')^{2(l-n-k) - l} \frac{ (2l-1-2k)!! }{k!}
  \frac{ (2n-1+2k)!! }{(l-n-k)!}.
\end{align}
This yields approximation \eqref{eq:GF_expressions}.
By the arguments in \aref{sec:causality_arguments}, this approximation is valid only if $r_* + r_*' - \Delta{t} \gg 1$.
Also, causality dictates that $G' \neq 0$ only for $\Delta{t} > \abs{r_* - r_*'}$.

One can show via algebraic manipulations that, upon replacing $r$ with $r_*$, the above expression for the Green's function is equivalent to the expression $G^{I}$ \eqref{eq:Barack_GII_expression} in ref.~\cite{Barack:1998bw} by Barack.
For example, we have, for $l=2$:
\begin{align}
  \frac12 G^{I}(u,v;u',v') &= \frac{3 r_*^4 + 2 r_*^2 ((r_*')^2 - 3 (\Delta{t})^2) + 3 ((r_*')^2 - (\Delta{t})^2)^2}{16 r_*^2 (r_*')^2} \qq{(Barack et al)} \nonumber \\
  G'(t,r_*;t',r_*') &=   \frac{3 r^4 + 2 r^2 ((r')^2 - 3 (\Delta{t})^2) + 3 ((r')^2 - (\Delta{t})^2)^2}{16 r^2 (r')^2} \qq{(this work)} \;.
\end{align}
This equivalence is not surprising, since Barack employ a $V \approx l(l+1)/r_*^2$ approximation to the Regge-Wheeler equation and solved for the corresponding Green's function.
Numerically, we found that our expression is slightly more accurate than that of Barack's, owning to a slightly more accurate approximation for $V(r_*)$.

\subsection{Approximations of $G_Q$ and $G_B$}
\label{sec:GQ_and_GB}
Various approximations of $G_Q$ and $G_B$ had been derived in the literature~\cite{Leaver:1986gd,Andersson:1996cm}.
Since we are most interested in the late time ($t-t' \gg r_*,r_*'$) and long range ($r_*, r_*' \gg 1$) regime in this work, we refer to the following well known approximations~\cite{Leaver:1986gd,Berti:2006wq}:
\begin{align}
  \label{eq:GQ_GB_expressions}
  G_Q(t,r_*;t',r_*') &\approx %\Theta(t-t'-r_*-r_*') \times
                       \Re \left[ \sum_n B_n e^{s_n (t-t'-r_*-r_*')} \right]
    \nonumber \\
  G_B(t,r_*;t',r_*') &\approx %\Theta(t-t'-r_*-r_*') \times
                       (-1)^{l+1} \frac{2(2l+2)!}{[(2l+1)!!]^2} (rr')^{l+1} (t-t')^{-2l-3} \;.
\end{align}
%where the expression for $G_B$ is only valid for $t-t' \gg r_*, r_*'$.
In $G_Q$, the sum is over poles of $G(s,r_*,r_*')$, $\omega_n = i s_n$'s are the quasinormal frequencies, and the $B_n$'s are quasinormal excitation factors. % $\omega_n = i s_n$
In $G_B$, the $(t-t')^{-2l-3}$ yields the well known Price's power law.

\section{Details of $\Psi$ calculation}
\label{sec:Psi_details}

\subsection{$\Psi_F$}
\label{sec:Psi_F_details}
Here are some examples for the $I_{l,\beta}(t,r_*)$ integral defined in \eqref{eq:Psi_F_integral_u_v}, along with their late times asymptotics.
We have assumed $w'=t_i-r_i=0$ in the arguments.
\begin{align}
  \label{eq:I_integral_examples}
  I_{0,0}(t,r_*) &= \int_{t-r}^{t+r} \frac{1}{4} \dd{z'} = \frac{r}{2} \sim \frac{r}{2} \nonumber \\
I_{0,1}(t,r_*) &= \int_{t-r}^{t+r} \frac{1}{2 (z')} \dd{z'} = \frac{1}{2} \log \left(-\frac{r+t}{r-t}\right) \sim \frac{r}{t} \nonumber \\
I_{0,2}(t,r_*) &= \int_{t-r}^{t+r} \frac{1}{(z')^2} \dd{z'} = -\frac{2 r}{r^2-t^2} \sim \frac{2 r}{t^2} \nonumber \\
I_{0,3}(t,r_*) &= \int_{t-r}^{t+r} \frac{2}{(z')^3} \dd{z'} = \frac{4 r t}{(r-t)^2 (r+t)^2} \sim \frac{4 r}{t^3} \nonumber \\
I_{1,0}(t,r_*) &= \int_{t-r}^{t+r} \frac{r^2+t ((z')-t)}{4 r (z')} \dd{z'} = \frac{\left(r^2-t^2\right) \tanh ^{-1}\left(\frac{r}{t}\right)+r t}{2 r} \sim \frac{r^2}{3 t} \nonumber \\
I_{1,1}(t,r_*) &= \int_{t-r}^{t+r} \frac{r^2+t ((z')-t)}{2 r (z')^2} \dd{z'} = \frac{t \tanh ^{-1}\left(\frac{r}{t}\right)}{r}-1 \sim \frac{r^2}{3 t^2} \nonumber \\
I_{1,2}(t,r_*) &= \int_{t-r}^{t+r} \frac{r^2+t ((z')-t)}{r (z')^3} \dd{z'} = 0 \sim 0 \nonumber \\
I_{1,3}(t,r_*) &= \int_{t-r}^{t+r} \frac{2 \left(r^2+t ((z')-t)\right)}{r (z')^4} \dd{z'} = -\frac{4 r^2}{3 (r-t)^2 (r+t)^2} \sim -\frac{4 r^2}{3 t^4}
\end{align}
See \tref{tab:nonlinear_tail_amplitudes} for more asymptotics.

\subsection{$\Psi_Q$}
\label{sec:Psi_Q_details}
With an argument similar to that in \sref{sec:Psi_analytical_approximation}, we find that the integration region relevant for $\Psi_Q$ is $t_i + r_i \leq v' \leq t - r_*$ and $-\sigma + (t_i-r_i) \leq u' \leq \sigma + (t_i-r_i)$:
\begin{align}
  \label{eq:Psi_Q_integral_u_v}
  \Psi_Q(t,r_*)
  = & \  \frac12 \int_{t_i+r_i}^{t-r_*} \int_{t_i-r_i-\sigma}^{t_i-r_i+\sigma} G_Q(t,r_*; t', r_*') Q(t',r_*') \dd{u'} \dd{v'} \;.
\end{align}
Moreover, $G_Q$ can be written in light cone coordinates as:
\begin{align}
  G_Q(t,r_*;t',r_*') &\approx %\Theta(t-v'-r_*) \times
                       \Re \left[ \sum_n B_n e^{s_n (t-v'-r_*)} \right] \;.
\end{align}
However, if we proceed with this calculation, we would find $\Psi_Q \sim t^{-\beta}$ at late times.
This would imply that $\Psi_Q$ can dominate over $\Psi_F$ at late times, which clearly contradicts the numerical results in \sref{sec:sourced_tail_numerical}.
We believe the above calculation of $\Psi_Q$ is invalid because the approximation we use for $G_Q$ is not accurate around $t - t' \approx r_* + r_*'$.
To provide further evidence for this claim, we refer to eq.~(34) of Barack \cite{Barack:1998bw}, which gives the quasinormal Green's function for a simplified model.
Barack's formula indicates that the quasinormal Green's function should not consist of merely exponentials, but also extra power law factors, which can modify the actual late time behavior of $\Psi_Q$.

\subsection{$\Psi_B$}
\label{sec:Psi_B_details}
For $\Psi_B$, we perform the integral in $r_*'$ and $t'$ coordinates.
Since $F$ is supported on $[-\sigma,\sigma]$, the integration boundary for $r_*'$ can be restricted to $t'-(t_i-r_i)-\sigma \leq r_*' \leq t'-(t_i-r_i)+\sigma$ for fixed $t'$:
\begin{align}
  \Psi_B(t,r_*) &= \int_{t_i}^t \int_{t'-(t_i-r_i)-\sigma}^{t'-(t_i-r_i)+\sigma}  G_B(t,r_*;t',r_*') \frac{F((t'-r_*')-(t_i-r_i))}{(r')^\beta} \dd{r_*'} \dd{t'}  \;.
\end{align}
Note that the only $r_*'$ dependence in the integrand are the powers $(r')^{l+1}$ and $(r')^{-\beta}$.
Assuming that $\sigma \ll r_*'$ always holds in the integration region, we can approximate $r'$ with $\eval{r'}_{r_*'=t'-(t_i-r_i)}$.
This leads to a significant simplification in the integral:
\begin{align}
  \Psi_B(t,r_*) &\approx (-1)^{l+1} \frac{2(2l+2)!}{[(2l+1)!!]^2} r^{l+1} t^{-2l-3}  K(t) \int_{-\sigma}^{\sigma} F(x) \dd{x} \qq{where} \nonumber \\
  K(t) &= \int_{t_i} \left(1-\frac{t'}{t}\right)^{-2l-3} \left(\eval{r'}_{r_*'=t'-(t_i-r_i)}\right)^{l+1-\beta}   \dd{t'} \;.
  %   & \int_{t'-(t_i-r_i)-\sigma}^{t'-(t_i-r_i)+\sigma} F((t'-r_*')-(t_i-r_i)) \dd{r_*'}\nonumber \\
% r'(r_*'=t'-(t_i-r_i))
\end{align}
Here we omit the upper boundary integral in $K(t)$ since the integral is divergent at $t' = t$.
This apparent divergence is due to the fact that expression \eqref{eq:GQ_GB_expressions} for $G_B$ is not valid for $t'\approx t$.
Nevertheless, we can proceed by assuming the integral is dominated by $t' \approx t_i$, where the source is the most intense.
In the case $l+1-\beta<-1$, the integral then converges, and we have:
\begin{align}
  K(t) &%\approx \int_{t_i} \left(\eval{r'}_{r_*'=t'-(t_i-r_i)}\right)^{l+1-\beta}   \dd{t'} \int_{-\sigma}^{\sigma} F(x) \dd{x}
         \approx \frac{-1}{l+2-\beta} \left(\eval{r'}_{r_*'=r_i}\right)^{l+2-\beta}  \;.
\end{align}
We see that $K(t)$ is approximately independent of $t$.
Although we cannot do the same calculation for $l+1-\beta \geq -1$, we still expect a $K(t)$ roughly constant in time, and $\Psi_B \sim t^{-2l-3}$ at late times.

\section{Source terms due to cubic coupling}
\label{sec:rsh_couplings}

Real spherical harmonics (RSH) $Y_{lm} : S^2 \to \mathbb{R}$ are defined in terms of spherical harmonics $Y_l^m : S^2 \to \mathbb{C}$ via
\begin{align}
  Y_{lm} =
  \begin{cases}
    \frac{i}{\sqrt{2}} (Y_l^m - (-1)^m Y_l^{-m}) & m < 0 \\
    Y_l^0 & m = 0 \\
    \frac{1}{\sqrt{2}} (Y_l^{-m} + (-1)^m Y_l^m) & m > 0 
  \end{cases}
  \;.
\end{align}
In \sref{sec:self_interacting_scalar_field}, we chose to expand $\Psi$ in terms of RSH to avoid the need for handling reality conditions.
In our setup, these RSH modes are coupled together via eq.~\eqref{eq:cubic_scalar_Psi_lm_eqn}, with the coupling terms acting as effective sources.
The full list of leading order source terms defined in \sref{sec:coupled_simulation_analytical} are given by:
\begin{align}
  \label{eq:Q_lm_leading_order}
  Q_{00}^{(1)} & = \frac{1-r}{r^2} \left( \frac{(\Psi^{(0)}_{11})^2}{2 \sqrt{\pi }} \right) \nonumber \\
  Q_{11}^{(2)} & = \frac{1-r}{r^2} \left( \frac{(\Psi^{(1)}_{00}) (\Psi^{(0)}_{11})}{\sqrt{\pi }}-\frac{(\Psi^{(1)}_{20}) (\Psi^{(0)}_{11})}{\sqrt{5 \pi }}+\sqrt{\frac{3}{5 \pi }} (\Psi^{(1)}_{22}) (\Psi^{(0)}_{11}) \right) \nonumber \\
  Q_{20}^{(1)} & = \frac{1-r}{r^2} \left( -\frac{(\Psi^{(0)}_{11})^2}{2 \sqrt{5 \pi }} \right) \nonumber \\
  Q_{22}^{(1)} &= \frac{1-r}{r^2} \left( \frac{1}{2} \sqrt{\frac{3}{5 \pi }} (\Psi^{(0)}_{11})^2 \right) \nonumber \\
  Q_{31}^{(2)} & = \frac{1-r}{r^2} \left( 3 \sqrt{\frac{2}{35 \pi }} (\Psi^{(0)}_{11}) (\Psi^{(1)}_{20})-\sqrt{\frac{3}{70 \pi }} (\Psi^{(0)}_{11}) (\Psi^{(1)}_{22}) \right) \nonumber \\
  Q_{33}^{(2)} & = \frac{1-r}{r^2} \left( \frac{3 (\Psi^{(0)}_{11}) (\Psi^{(1)}_{22})}{\sqrt{14 \pi }} \right) \nonumber \\
  Q_{40}^{(3)} & = \frac{1-r}{r^2} \left( \frac{3 (\Psi^{(1)}_{20})^2}{7 \sqrt{\pi }}+\frac{(\Psi^{(1)}_{22})^2}{14 \sqrt{\pi }}-\sqrt{\frac{2}{7 \pi }} (\Psi^{(0)}_{11}) (\Psi^{(2)}_{31}) \right) \nonumber \\
  Q_{42}^{(3)} & = \frac{1-r}{r^2} \left( \frac{1}{7} \sqrt{\frac{15}{\pi }} (\Psi^{(1)}_{20}) (\Psi^{(1)}_{22})+\sqrt{\frac{5}{14 \pi }} (\Psi^{(0)}_{11}) (\Psi^{(2)}_{31})-\frac{(\Psi^{(0)}_{11}) (\Psi^{(2)}_{33})}{\sqrt{42 \pi }} \right) \nonumber \\
  Q_{44}^{(3)} & = \frac{1-r}{r^2} \left( \frac{1}{2} \sqrt{\frac{5}{7 \pi }} (\Psi^{(1)}_{22})^2+\sqrt{\frac{2}{3 \pi }} (\Psi^{(0)}_{11}) (\Psi^{(2)}_{33}) \right) \;.
\end{align}

\section{Relation between the Regge-Wheeler and the Zerilli Green's function}
\label{sec:darboux_transform}
It is well known that the Regge-Wheeler equation and the Zerilli equation are supersymmetric partners of each other \footnote{Similar to that of supersymmetric quantum mechanics \cite{Cooper:1994eh}.}, which is also the reason why the parity odd and even perturbations of a Schwarschild BH are {\it isospectral} \cite{Chandrasekhar:1975zza}. Obviously, the branch cuts of the Green's functions in both sector are related to each other. In this section, we utilize this relation and show that one can derive the Zerilli Green's function from the Regge-Wheeler Green's function, or the other way around. 

In Laplace space, the Regge-Wheeler/Zerilli equation for the parity odd/even sector ($\phi_+$/$\phi_-$) 
\begin{align} \label{eqn:laplRWZeqn}
   \left( s^2 -\partial_{r^*}^2 + V_{\rm RW/Z}\right)\phi_{\rm RW/Z} =0
\end{align}
can be written in the form 
\begin{align}
   &\left( s^2 - \beta + A A^\dagger \right) \phi_{RW} =0, \nonumber \\
   & \left( s^2 - \beta + A^\dagger A \right) \phi_{Z} =0,
\end{align}
where $\beta = \frac{4\lambda^2(\lambda+1)^2}{9r_s^2}$ and 
\begin{align}
   & A = \partial_{r^*}+W(r^*), \quad  A^\dagger = -\partial_{r^*}+W(r^*), \nonumber \\
   & W(r) = \frac{3 r_s(r_s-r) }{r^2 (3r_s + 2\lambda r)} - \frac{2\lambda(\lambda+1)}{3r_s} , \quad \lambda= \frac{(\ell -1)(\ell+2)}{2}.
\end{align}
Therefore it is straightforward to recognizes that the solutions for the parity even and odd variables are proportional to the Darboux transformation of each other 
\begin{align} \label{eqn:susytrans}
    \phi_{\rm Z} \propto A^\dagger \phi_{\rm RW}, \quad \phi_{\rm RW} \propto A \phi_{\rm Z} ,
\end{align}
except the algebraically special mode \cite{Chandrasekhar:1984mgh}. The Wronskian transforms as 
\begin{align}
    W[A^\dagger \phi_-, A^\dagger \phi_+] = (\beta -s^2)W[\phi_-,  \phi_+] .
\end{align}
with 
\begin{align}
    \phi_- & \sim e^{-sr^*}  \quad   {\rm for } \quad  r^* \to \infty , \nonumber \\
    \phi_+ &\sim e^{sr^*}  \quad   {\rm for } \quad  r^* \to -\infty .
\end{align}
Obviously these asymptotics are invariant under the transformation~\eqref{eqn:susytrans}. With a bit of algebra, one find that the Green's function satisfy the following identity
\begin{align}
    \tilde{G}_s(x,y) = \frac{1}{\beta-s^2}[A^\dagger_x A^\dagger_y G_s(x,y)-\delta(x-y)],
\end{align}
where
\begin{align}
    \tilde{G}_s(x,y) = \frac{\tilde{\phi}_-(x)\tilde{\phi}_+(y)\theta(x-y)+\tilde{\phi}_-(y)\tilde{\phi}_+(x)\theta(y-x)}{W[\tilde{\phi}_-,\tilde{\phi}_+]}, \quad 
    \tilde{\phi}_{\pm}(x) = A^{\dagger}_x\phi_{\pm}(x) . 
\end{align}
It is worth commenting on the new poles $s^2 =\beta$. It is possible that the original Green's function with $\phi_{\pm}$ does not have any pole at  $s^2 =\beta$, therefore the two problems are not completely isospectral. When one apply this relation between Green's functions, with $\phi \leftrightarrow \tilde{\phi}$  and $A^\dagger \leftrightarrow A$, as $G$ has no pole at $s^2 =\beta$, that means the Darboux transformation $A$ has to annihilate the solution $\tilde{\phi}$ at  $s^2 =\beta$. This is also why it is an algebraically special mode.
 
Given this relation, one can directly apply it to the Green's functions of~\eqref{eqn:laplRWZeqn} without the subtleties of the algebraically special mode. In time domain, the above relation leads becomes
\begin{align}
    \tilde{G}(t-t',x;y) = -A^\dagger_x  A^\dagger_y \int_0^{t-t'}\rd \tau \,\frac{\sinh(\sqrt{\beta} \,\tau )}{\sqrt{\beta}}G(t-t'-\tau,x;y) + \frac{\sinh(\sqrt{\beta} \,\tau )}{\sqrt{\beta}}\theta(t-t')\delta(x-y) .
\end{align}

\bibliography{main}

\end{document}